\documentclass[aps,prx,superscriptaddress,amsfonts,amsmath,amssymb,showpacs,floatfix,reprint,longbibliography]{revtex4-2}

\usepackage{physics}
\usepackage{amsfonts}
\usepackage{amssymb}
\usepackage{amsthm}
\usepackage{amsmath}
\usepackage{hyperref}
\usepackage{cleveref}
\usepackage{mathtools}
\usepackage{bbm}
\usepackage[dvipsnames]{xcolor}
\usepackage{bm}
\newcommand{\llangle}{\langle\!\langle}
\newcommand{\rrangle}{\rangle\!\rangle}

\newcommand{\NN}{\mathcal{N}}

\newcommand{\FF}{\mathcal{F}}

\newcommand{\LL}{\mathcal{L}}

\newcommand{\TT}{\mathcal{T}}
\newcommand{\ZZ}{\mathcal{Z}}
\newcommand{\MM}{\mathcal{M}}

\newcommand{\BB}{\mathcal{B}}

\DeclarePairedDelimiterX{\inner}[2]{\langle}{\rangle}{#1|#2}
\DeclarePairedDelimiterX{\expect}[3]{\langle}{\rangle}{#1|#2|#3}

\DeclareMathOperator*{\poly}{poly}

\definecolor{myblue}{RGB}{20, 60, 180}

\definecolor{col}{RGB}{250, 64, 47}
\hypersetup{colorlinks=true, linkcolor=col, citecolor=col, urlcolor=col}

\newtheorem{theorem}{Theorem}[section]
\newtheorem{lemma}{Lemma}[section]
\newtheorem{corollary}{Corollary}[section]
\newtheorem{prop}{Proposition}[section]
\newtheorem{defn}{Definition}[section]

\usepackage{algorithm}
\usepackage{algpseudocode}
\algnewcommand{\procname}[1]{\textnormal{\textsc{#1}}}

\hypersetup{
    colorlinks=true,
    linkcolor=magenta,
    citecolor=magenta,
    urlcolor=magenta
}

\begin{document}
\title{Beyond Belief Propagation: Cluster-Corrected Tensor Network Contraction with Exponential Convergence}

\author{Siddhant Midha}
\email{sm7456@princeton.edu}
\altaffiliation{equal contribution}
\affiliation{Princeton Quantum Initiative, Princeton University, Princeton, NJ 08544}

\author{Yifan F. Zhang}
\email{yz4281@princeton.edu}
\altaffiliation{equal contribution}
\affiliation{Department of Electrical and Computer Engineering, Princeton University, Princeton, NJ 08544}

\begin{abstract}
Tensor network contraction on arbitrary graphs is a fundamental computational challenge with applications ranging from quantum simulation to error correction. While belief propagation (BP) provides a powerful approximation algorithm for this task, its accuracy limitations are poorly understood and systematic improvements remain elusive. Here, we develop a rigorous theoretical framework for BP in tensor networks, leveraging insights from statistical mechanics to devise a \textit{cluster expansion} that systematically improves the BP approximation. We prove that the cluster expansion converges exponentially fast if an object called the \emph{loop contribution} decays sufficiently fast with the loop size, giving a rigorous error bound on BP. We also provide a simple and efficient algorithm to compute the cluster expansion to arbitrary order. We demonstrate the efficacy of our method on the two-dimensional Ising model, where we find that our method significantly improves upon BP and existing corrective algorithms such as loop series expansion. Our work opens the door to a systematic theory of BP for tensor networks and its applications in decoding quantum error-correcting codes and simulating quantum systems.

\end{abstract}

\maketitle
\tableofcontents

\section{Introduction}

Tensor networks (TNs) comprise a set of powerful mathematical and computational tools widely used in condensed matter physics and quantum information science. Originally developed for quantum many-body systems~\cite{white1992density,white1993density}, tensor networks have evolved into a unifying framework that connects diverse areas of physics~\cite{schollwock2011density,orus2014practical,cirac2021matrix,verstraete2004matrix,vidal2007entanglement} and computer science~\cite{Schuch2007PEPS,Schuch2007,Landau2015,MarkovShi2008,Arad2013}.

Despite their conceptual elegance and broad applicability, the practical utility of tensor networks is fundamentally limited by the computational complexity of tensor contraction. Contracting a tensor network—--summing over all internal indices to compute the final result--—is central to extracting physical quantities from the network representation. However, this operation generally requires exponential runtime in the system size~\cite{Schuch2007PEPS,HaferkampPEPS,HarleyPEPS}, motivating the development of polynomial-time algorithms for approximate network contraction.

Many such algorithms have been developed, often exploiting special structures in the tensor network to simplify contractions. For example, time evolving block decimation (TEBD) \cite{vidal2003efficient,Vidal2004} leverages the bounded growth of entanglement to truncate the bond dimension. Other algorithms exploit the network geometry; in particular, for geometries without loops (such as one-dimensional matrix product states~\cite{Fannes1992,PerezGarcia2007} and tree tensor networks~\cite{Shi2006}), exact contraction becomes efficient if performed in the right order.

Belief propagation (BP) \cite{pearl1988probabilistic, bethe1935statistical, kirkley2021belief, laumann2008cavity, chertkov2008exactness}, a classical algorithm rooted in computer science and statistical physics, has recently emerged as a promising candidate for approximate contractions of tensor networks~\cite{leifer2008quantum,robeva2016duality}. Originally developed by Pearl for probabilistic inference in graphical models~\cite{pearl1988probabilistic}, BP found its theoretical foundation in Bethe's work on lattice statistical mechanics~\cite{bethe1935statistical}. BP has since proven useful in tasks such as decoding classical and quantum low-density parity check (LDPC) codes~\cite{gallager1962low, mackay1999good, yao2024belief, koutsioumpas2025automorphism, mceliece1998turbo, old2023generalized,roffe2020decoding}, machine learning~\cite{yedidia2003understanding, wainwright2008graphical}, and optimization~\cite{guerreschi2020quantum, bravyi2018simulation}. While BP was developed in the context of classical probability theory, it is equally applicable to quantum systems and has been widely adopted~\cite{BP_gauging,BPsimpleupdate,blockBPpeps,blockBPdecoder, poulin2008belief, poulin2008iterative, sahu2022efficient, hastings2007quantum, chen2022tensor, wang2024tensor}. Notably, BP-based techniques have been used to classically simulate major quantum experiments, challenging claims of quantum advantage~\cite{tindall2025dynamics, rudolph2025simulatingsamplingquantumcircuits}.

BP offers several advantages over other methods: it is polynomial time and parallelizable, applies to arbitrary geometries, and becomes exact on networks without loops. However, the conditions under which BP is a `good' approximation are still not clear, especially in loopy geometries. In addition, unlike methods such as TEBD, which can be systematically improved (e.g., by increasing the bond dimension), the traditional BP method is inherently rigid and lacks a tuning parameter to trade off computational resources for lower error rates.

In this work, we develop a systematic theory of belief propagation (BP) for tensor network contractions to address these challenges. Our main contributions are: (1) establishing rigorous control over the difference between the exact value $Z$ and the BP contraction $Z_0$, and (2) designing an efficient algorithm that systematically corrects the BP error with exponential convergence.

We achieve our results through a novel approach that overcomes fundamental limitations of existing methods. In Ref.~\cite{evenbly2024loopseriesexpansionstensor,park2025simulatingquantumdynamicstwodimensional}, the authors take inspirations from earlier work in statistical inference \cite{chertkov2006loop,chertkov2006loop2, chertkov2009approximate,parisi2006loop, montanari2005compute, mooij2012sufficient, gomez2007truncating, ramezanpour2015statistical} to construct a Taylor series known as the \emph{loop expansion} in order to correct the BP approximation value toward the ground truth. The na\"{\i}ve loop expansion starts from the BP value $Z_0$ at zeroth order and sums corrections called \emph{loop tensors} that grow larger at higher orders. While this seems like a natural tuning knob—--when fully summed, the loop expansion yields $Z$—--the method suffers from a critical flaw: it does not generally converge, even when individual loop tensors decay exponentially. The fundamental problem is that the expansion includes disconnected loops whose number grows combinatorially with size, overwhelming any exponential decay and causing divergence. The work \cite{evenbly2024loopseriesexpansionstensor} made the first step in mitigating the combinatorial growth through multiple self-consistent methods. 

To resolve this convergence failure, we leverage insights from statistical mechanics to construct an entirely different series: a \emph{cluster expansion} that converges to the logarithm of the ground truth, $\log(Z)$, rather than $Z$ itself. Starting from $\log(Z_0)$, our method sums modified corrections called \emph{clusters}, derived from loop tensors. Crucially, only \emph{connected clusters} contribute to our expansion, and their number grows at most exponentially—--not combinatorially. This fundamental difference ensures that our cluster expansion converges exponentially fast, solving the convergence problem that plagues the loop series method.

We rigorously prove the exponential convergence of our cluster expansion, provided that loop tensors decay exponentially with a sufficiently large exponent. This resolves the two main challenges of BP: first, it explains when BP provides a good approximation and supplies rigorous error estimates; second, it yields a polynomial-time algorithm that systematically and reliably improves BP results—--something that loop series expansions cannot achieve due to their inherent convergence problems.

To understand when loops decay and the cluster expansion converges, we conduct extensive numerical experiments computing the free energy density of the 2D Ising model and benchmark against the exact solution. We observe distinctive behaviors across the phase transition: while BP performs well deep in the high-temperature and low-temperature phases, it deviates from the ground truth because of the presence of long-ranged fluctuations at the critical point. There, we show that cluster expansion significantly improves the BP error and converges faster than loop expansion, independent of the system size.

The change from the loop expansion to the cluster expansion is physically intuitive. Borrowing wisdom from statistical mechanics, we interpret tensor networks as partition functions. As an object that changes multiplicatively with system sizes, it is inherently unstable and is hard to approximate with a series expansion. In contrast, the logarithm of the partition function, or the free energy, is an extensive quantity that changes additively with system sizes. It has an additive response to local perturbation and can thus be approximated with a series expansion. This intuition is rigorized in our convergence proof and is supported by our numerical experiments.

\subsection{Comparison with existing and concurrent work}\label{sec:comparison}
We compare our results with earlier and concurrent works. Ref. \cite{chertkov2006loop,chertkov2006loop2} first introduce the loop series expansion in probabilistic graphical models. Ref. \cite{evenbly2024loopseriesexpansionstensor} generalizes the loop series expansion to tensor networks. Notably, they propose several methods to overcome the divergence of the na\"{\i}ve loop expansion. Their key intuition is the extensivity of free energy which also motivates our paper. In particular, the combinatorial toolbox we will use formalizes the counting argument underlying their ``single-excitation'' and ``multi-excitation'' approximation. Our results advance this line of work in three key respects. First, the cluster expansion yields explicit expressions for correction terms at all orders. Second, it applies beyond the thermodynamic limit, whereas the previous methods are restricted to that setting. Lastly, we resolve an open question by showing that sufficiently fast loop decay guarantees exponential convergence of the cluster expansion.

There are also earlier attempts to derive modified series expansion with better convergence than the na\"{\i}ve loop expansion.  Notable examples include the \emph{linked cluster} expansion \cite{rigol2006numerical}, also known as the \emph{cluster-cumulant} expansion \cite{welling2012cluster} in the statistical inference literature. Ref. \cite{park2025simulatingquantumdynamicstwodimensional} employs the linked cluster expansion as one of the subroutines to compute observables in time-evolved quantum systems. Concurrent to our work, Ref. \cite{gray2025tensor} appeared which employs the linked cluster expansion to perform extensive numerical simulations of quantum ground states. We note that while the linked cluster expansion is conceptually very similar to the cluster expansion derived in our work, they are formally two different series. See Section \ref{sec:toy} for an example. However, we believe that our theory of convergence should apply to linked cluster expansion as well, since both series reflect fundamentally the same physics.


\subsection{Organization}
This paper is organized as follows. In Section~\ref{sec:bp}, we review the belief propagation algorithm and its connection to tensor networks. In Section~\ref{sec:cluster}, we introduce the loop and cluster expansions, present our main convergence theorem, and discuss its implications. In Section~\ref{sec:algorithm}, we outline the algorithmic procedure for computing the cluster expansion. In Section~\ref{sec:ising}, we present numerical results on the two-dimensional Ising model where the exact solution is known. Finally, in Section~\ref{sec:discussion}, we discuss potential extensions and applications of our work.

\begin{figure*}[t]   
    \centering
    \includegraphics[width=1.0\linewidth]{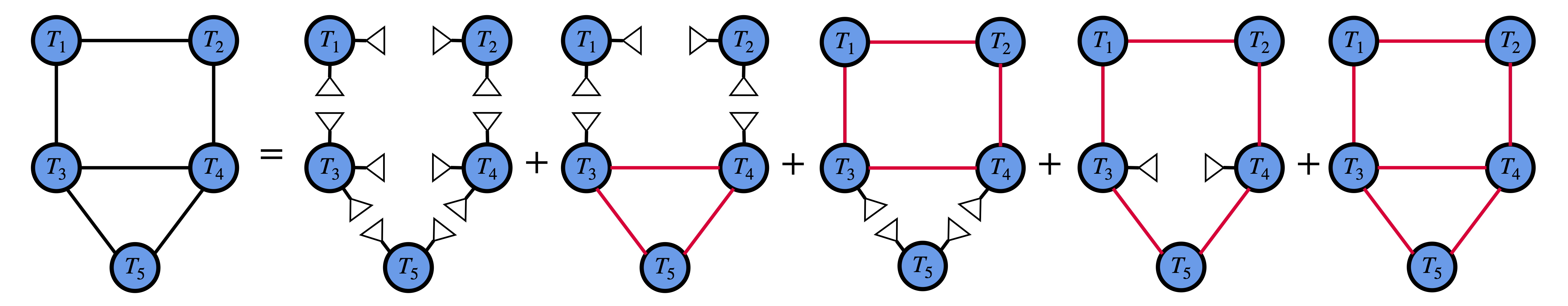}  
    \caption{\textbf{Loop series expansion}. The contraction of a five-vertex tensor network can be exactly represented as the sum of the BP vacuum and all the generalized loop excitations on the graph.}
    \label{fig:fig_bp}  
\end{figure*}

\section{Belief Propagation}\label{sec:bp}
In this section, we define our notation and introduce the BP algorithm. We then introduce the loops series expansion and show why it does not converge in general.

\subsection{Tensor networks}
We consider a tensor network, with no open indices, defined on a graph $G = (V,E)$ on $N$ sites with sets of vertices $V$ and edges $E$. For each vertex $v \in V$, We denote the neighbors of $v$ in $G$ as $\NN(v)$. The degree of a vertex in the graph is denoted $d(v) := |\NN(v)|$. The degree of a graph is denoted $\Delta(G) := \max_{v\in V}d(v)$. Next, we define the notion of a tensor network on the graph. For edge $(v,w)\in E$ we associate the \textit{bond} Hilbert space $\BB_{vw}$ with $\text{dim}(\BB_{vw})$ the \textit{bond dimension} on that edge, which we take to be uniform and equal to $\chi$ without loss of generality. Each vertex $v\in V$ is then equipped with a tensor $T_v \in \otimes_{n \in \NN(v)}\BB_{nv}$.

We refer to the triple $\TT = (\{T_v\}_{v\in V}, V, E)$ as a \textit{tensor network}. These networks have no `open' indices, and will be used for the belief propagation algorithm, as detailed in the following. Defined this way, the contraction of all the tensors in the network is a scalar, 
\begin{equation}
    \ZZ(\TT)= \star_{v\in V}T_v,
\end{equation}
where $\star$ denotes contraction of tensor indices. Following the convention in statistical mechanics, $\ZZ(\TT)$ is a sum of local terms over local Hilbert spaces, which is what usually defines a partition function in statistical mechanics. In fact, any partition function of a local statistical mechanical model can be written as a tensor network inheriting the locality of the model. Generalizing this language, we refer to $\ZZ(\TT)$ the \emph{partition function} of the tensor network. Note that $\ZZ(\TT)$ is strictly more general than the partition function in statistical mechanics: in general, tensors $T_v$ can be complex, while in statistical mechanics we only sum over non-negative quantities.

Following the same intuition, we will define the \emph{free energy} as the negative logarithm of the partition function, generalizing the definition in statistical mechanics.
\begin{equation}
    \FF(\TT) = -\log{\ZZ} 
\end{equation}
One might worry that $\ZZ$ can be complex and thus rendering $\FF(\TT)$ multi-valued. In the following discussion, we will always choose a normalization such that $\ZZ$ is positive and store the phase information separately. This ensures the uniqueness of $\FF(\TT)$ in the neighborhood where cluster expansion operates.

\subsection{Belief propagation}
Now we introduce the belief-propagation procedure. It begins with defining \textit{message tensors} on each edge of the graph, with $\mu_{v \to w} \in \BB_{vw}$ denoting the message from node $v$ to $w$. We define a set of fixed-point messages on the network through the notion of \textit{self-consistency} between the tensors and the messages. This requires that the contraction of all but one incoming message on any vertex $v\in V$ must result in the outgoing message on that edge. Mathematically, the self consistent set $\MM = \{\mu_{v\to w}\}_{v,w}$ satisfies for each $v\in V$ and each $s \in \NN(v)$,

\begin{equation}
        \left(\underset{n_i \in \NN(v)/\{n_j\}}{\otimes}\mu_{n_i \to v} \right)\star T_v = \mu_{v\to n_j}
\end{equation}
\noindent 
Schematically,
\begin{figure}[H]
    \centering
    \includegraphics[width=5cm, keepaspectratio]{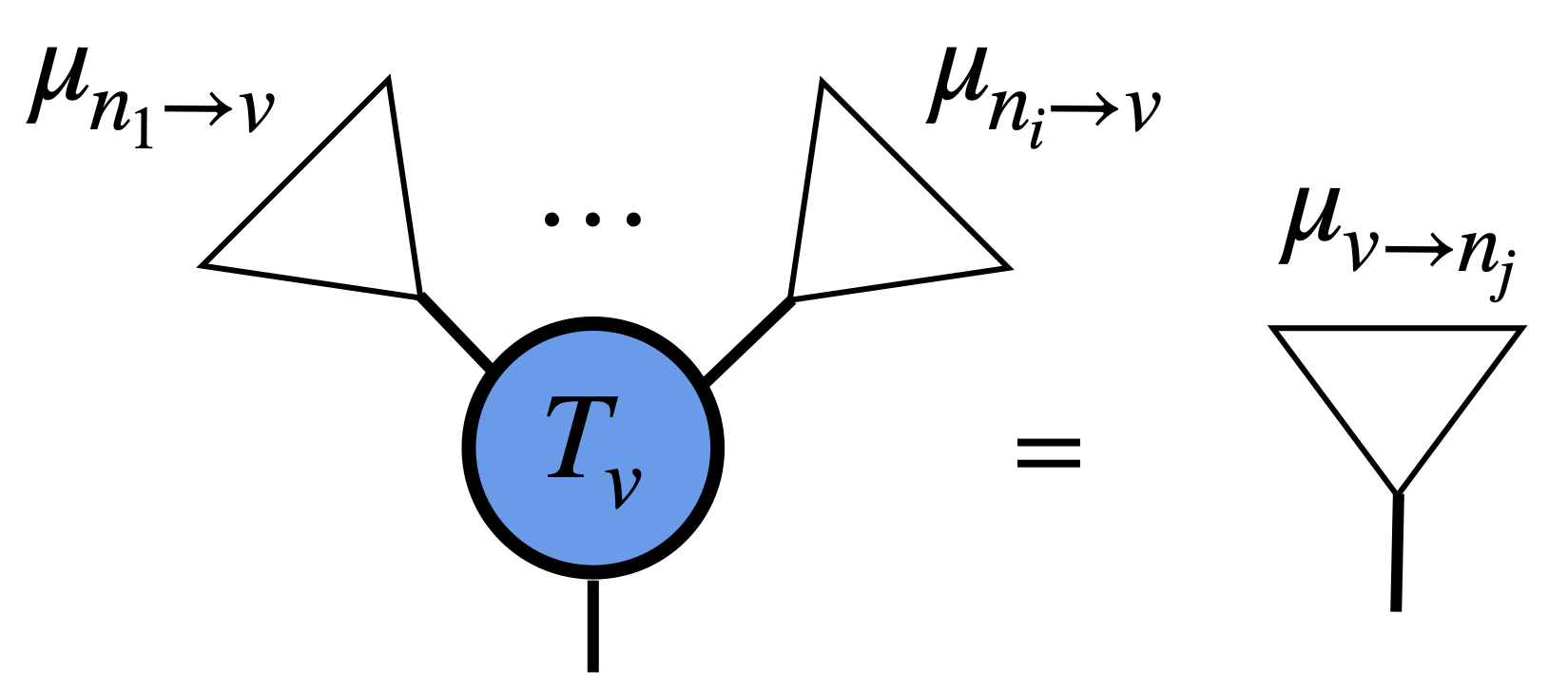} 
\end{figure}
where rank-one $\triangleright$ denote message tensors and $T_v$ is the local tensor at vertex $v$.

To evaluate the BP approximation, we first compute $I_{vw} = \mu_{v\to w} \star \mu_{w\to v}$, that is the inner product between $\mu_{v\to w}$ and $\mu_{w\to v}$ \footnote{note that no complex conjugation is take here, in other words $I_{vw}=I_{wv}$}.
For each $v\in V$, the local contribution to the partition function is then given as,
\begin{equation}\label{eq:tensor_z}
        Z^{(v)} := \left[\underset{n \in \NN(v)}{\bigotimes}\frac{\mu_{n\to v}}{\sqrt{I_{vw}}}\right] \star T_v,
\end{equation}
where we take the principle branch of the square root when $I_{vw}$ is complex. The BP vacuum solution to the partition function and the free energy are then given as, 

\begin{align}
    Z_0 &= \prod_{v\in V} Z^{(v)} \\ 
    F_0 &= -\sum_{v\in V}\log{Z^{(v)}}
\end{align}
This is graphically shown below. $Z_0$ is the BP approximation to $\ZZ$.
\begin{figure}[H]
    \centering
    \includegraphics[width=7cm, keepaspectratio]{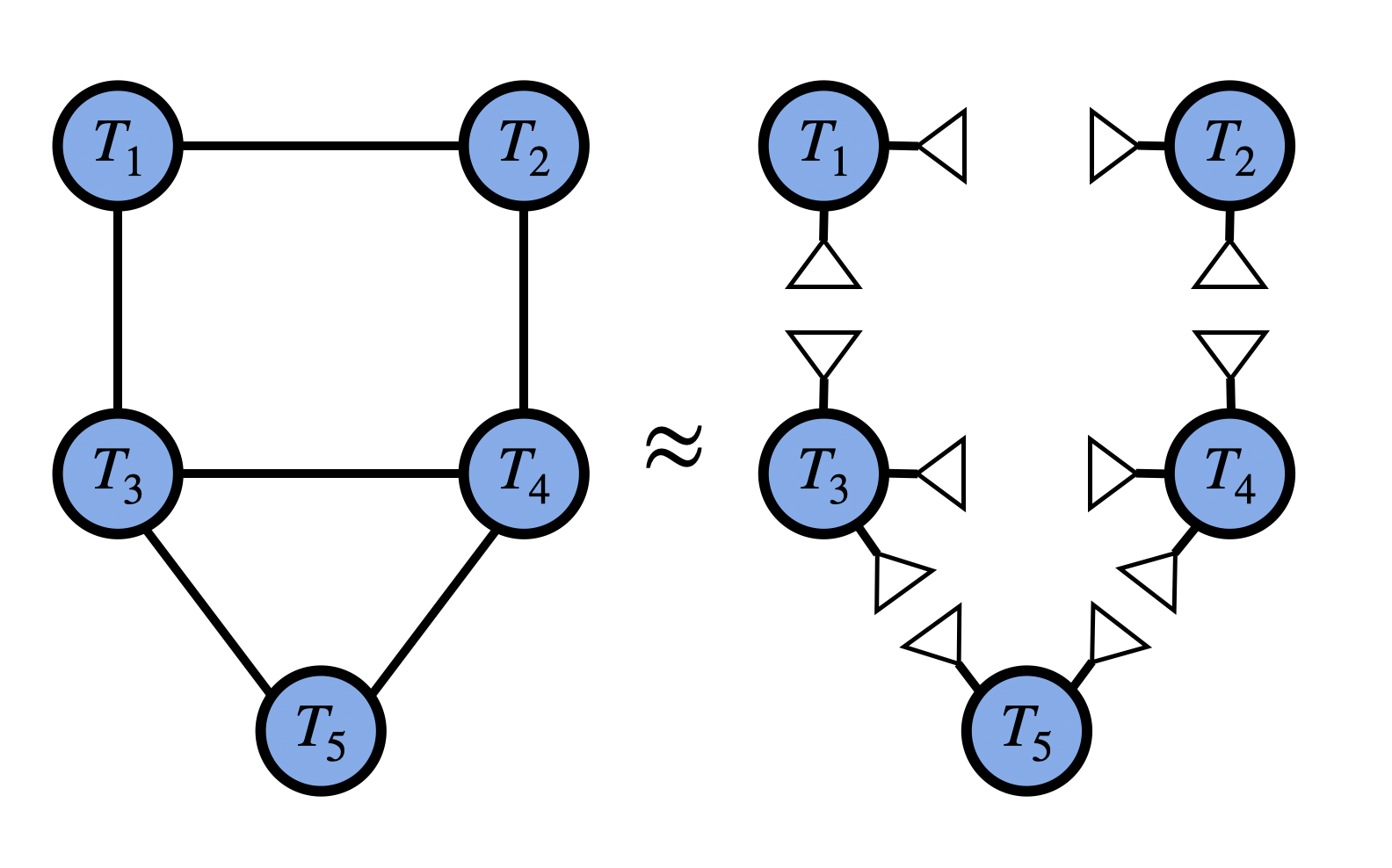} 
\end{figure}

Viewed as a mean-field approximation, the `traditional' BP algorithm is used to approximate the exact partition function with the vacuum solution $Z_0$. Physically, the message tensors can be viewed as a rank-one approximation of the influence of a complex environment. The local reduced density matrices are then given by contracting the rank-one environments into the local tensors as follows,

\begin{figure}[H]
    \centering
    \includegraphics[width=6cm, height = 4cm]{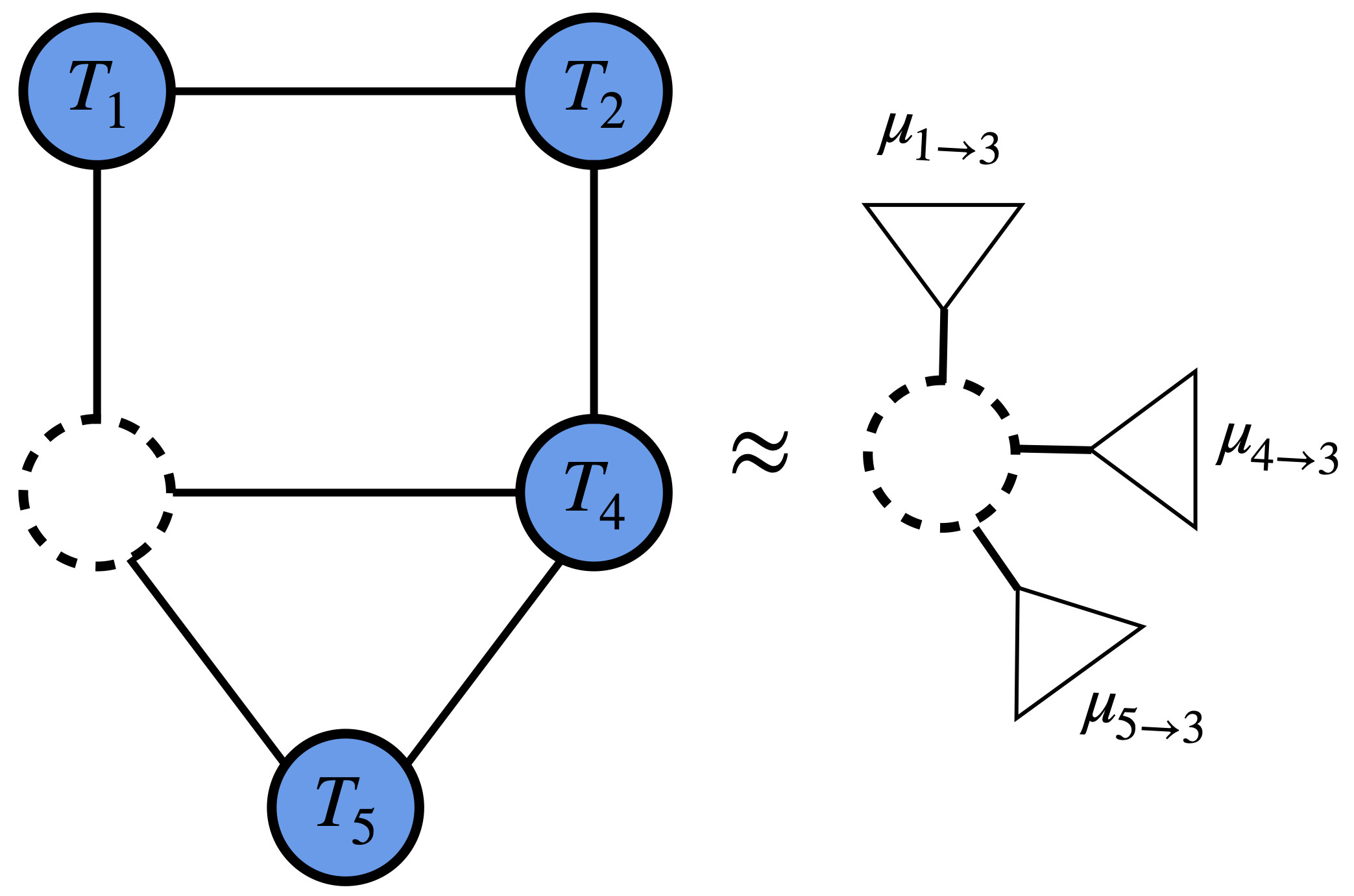} 
\end{figure}

\subsection{Loop series expansion}

The \textit{loop series expansion} \cite{evenbly2024loopseriesexpansionstensor} for the self-consistent messages in terms of the `generalized loops' on the graph can be written as follows. For each edge $e = (r,s)\in E$, consider the identity $\mathbbm{1} \in \LL(\BB_e)$. We define the orthogonal projector $\mathcal{P}_{rs}^\perp$ by expanding the identity as,
\begin{eqnarray}
    \mathbbm{1} = \frac{|\mu_{r\to s}\rrangle \llangle \mu_{s\to r}|}{I_{rs}} + \mathcal{P}_{rs}^\perp 
\end{eqnarray}
This ensures that $\mathcal{P}_{rs}^\perp|\mu_{r\to s}\rrangle = 0$ and $\llangle \mu_{s\to r}|\mathcal{P}_{rs}^\perp = 0$, hence $\mathcal{P}_{rs}^\perp$ carries contributions orthogonal to the BP vacuum. This is shown in Fig.~\ref{fig:fig_bp}(b), 

Now, consider the problem of evaluating the partition function $\ZZ$. Inserting the identity above at each edge in the network, one obtains $2^{|E|}$ terms. Each term can be expressed by an $|E|-$bit string $s: E \to \{0,1\}$, where $s(e)=0$ for edge $e =(r,s)\in E$ represents the BP vacuum contribution from $|\mu_{r\to s}\rrangle \llangle \mu_{s\to r}|$ and $s_e = 1$ represents the `excited edge' contribution from the orthogonal projector $\mathcal{P}_{rs}^\perp$. Each configuration $s$ defines an edge-induced subgraph \footnote{For a graph $(V, E)$, any subset of edges $F \subseteq E$ defines the \emph{edge-induced subgraph} $G[F] = (V_F, F)$, where $V_F$ is the set of vertices that are incident to edges in $F$.} $G_s \subset G$ of excited edges. This results in, 

\begin{equation}
    \ZZ = Z_0 + \sum_{s \neq \mathbf{0}}Z_s 
\end{equation}
where $Z_s$ denotes the contribution from configuration $s$ normalized by the vacuum contribution. Now, we note the that any configuration $s$ which has an `open' edge will vanish. The `vacuum' contribution is $Z_0$.

\begin{defn}[Generalized loops]\label{def:loop}
    Consider a graph $G=(V,E)$. A generalized loop is subgraph $C=(W,F)$ with $W \subseteq V$, $F \subseteq E$, with the property that the degree of any $w \in W$ in $C$ is at least two. The weight of a generalized loop is defined as the number of edges $|F|$.
\end{defn}

We denote the set of generalized loops in graph $G$ as $\LL_G$. Note that a generalized loop need not be a simple loop or even a connected subgraph. With mild abuse of notation, we refer to generalized loops simply as ``loops" and specify ``simple loops" when needed. We denote a loop as $l \in \LL_G$ with loop weight $|l|$.

\begin{lemma}
A non-zero excitation $Z_s$ is possible only if $G_s$ is a generalized loop in $G$ \cite{evenbly2024loopseriesexpansionstensor,chertkov2006loop}.
\end{lemma}

Hence, contributions such as the following with a ``dangling excitation'' vanish,
\begin{figure}[H]
    \centering
    \includegraphics[width=4cm, keepaspectratio]{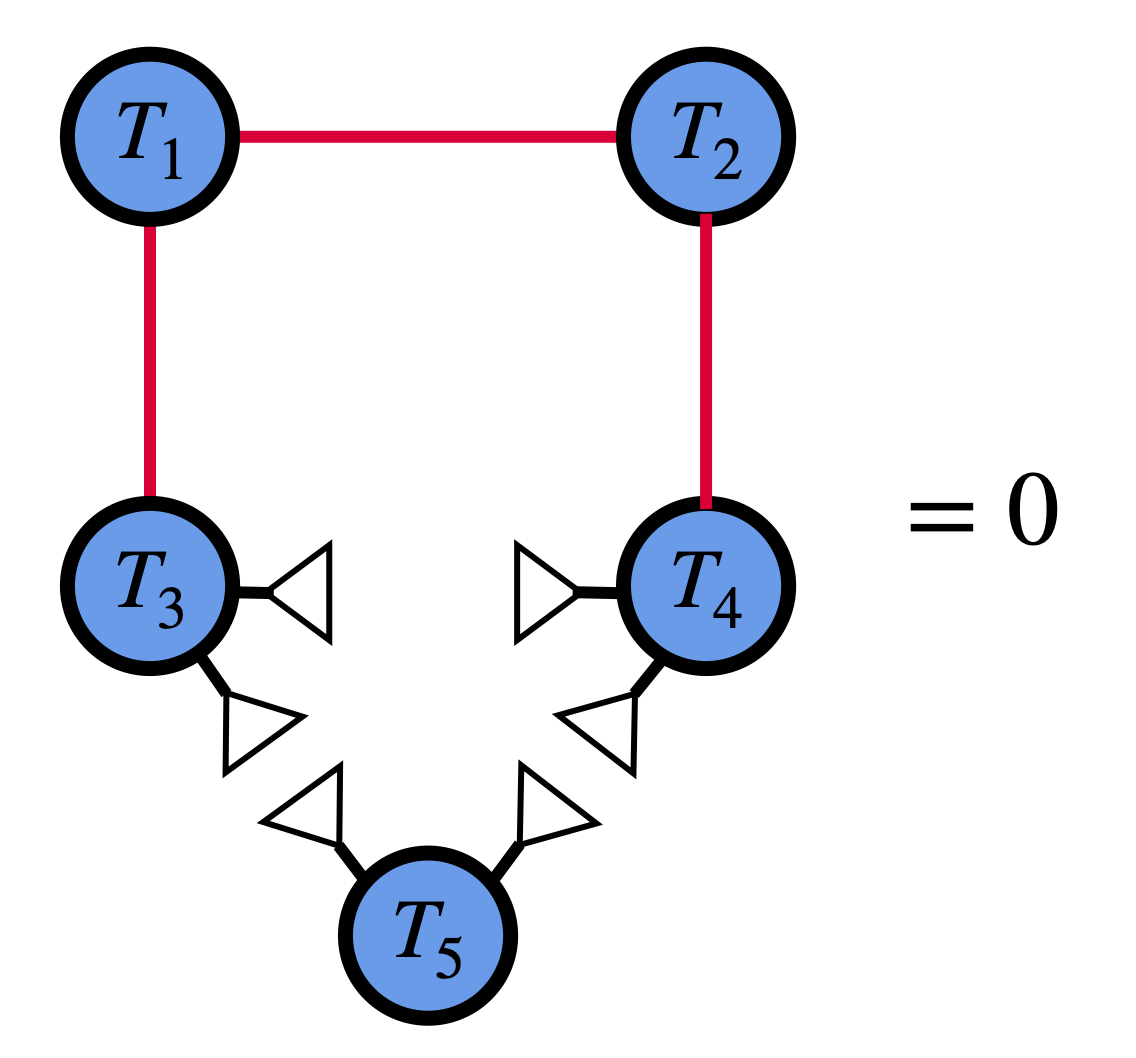} 
\end{figure}

Thus, it is possible to write the series expansion involving only those configurations which form generalized loops. Let us first define formally a \textit{loop correction}.

\begin{defn}[Loop correction]
Let \(l\) be a generalized loop in the graph. Each loop \(l\) has an associated loop correction \(Z_l\in\mathbb{C}\), defined as
\begin{equation}
  Z_l = \left( \prod_{(w,v)\in V(l)} \mathcal{P}_{wv}^\perp \right) \star \left( \prod_{(w,v) \notin V(l)} \mu_{wv} \otimes \mu_{vw} \right) \star \left( \prod_{v} T_v \right)
\end{equation}
Where \(V(l)\) is the set of vertices in the loop \(l\).
\end{defn}
For instance, a simple loop correction on the loop $l = \{(1,2),(2,3),(3,4),(1,4)\}$ on a square lattice is,  

\begin{figure}[H]
    \centering
    \includegraphics[width=3.8cm, keepaspectratio]{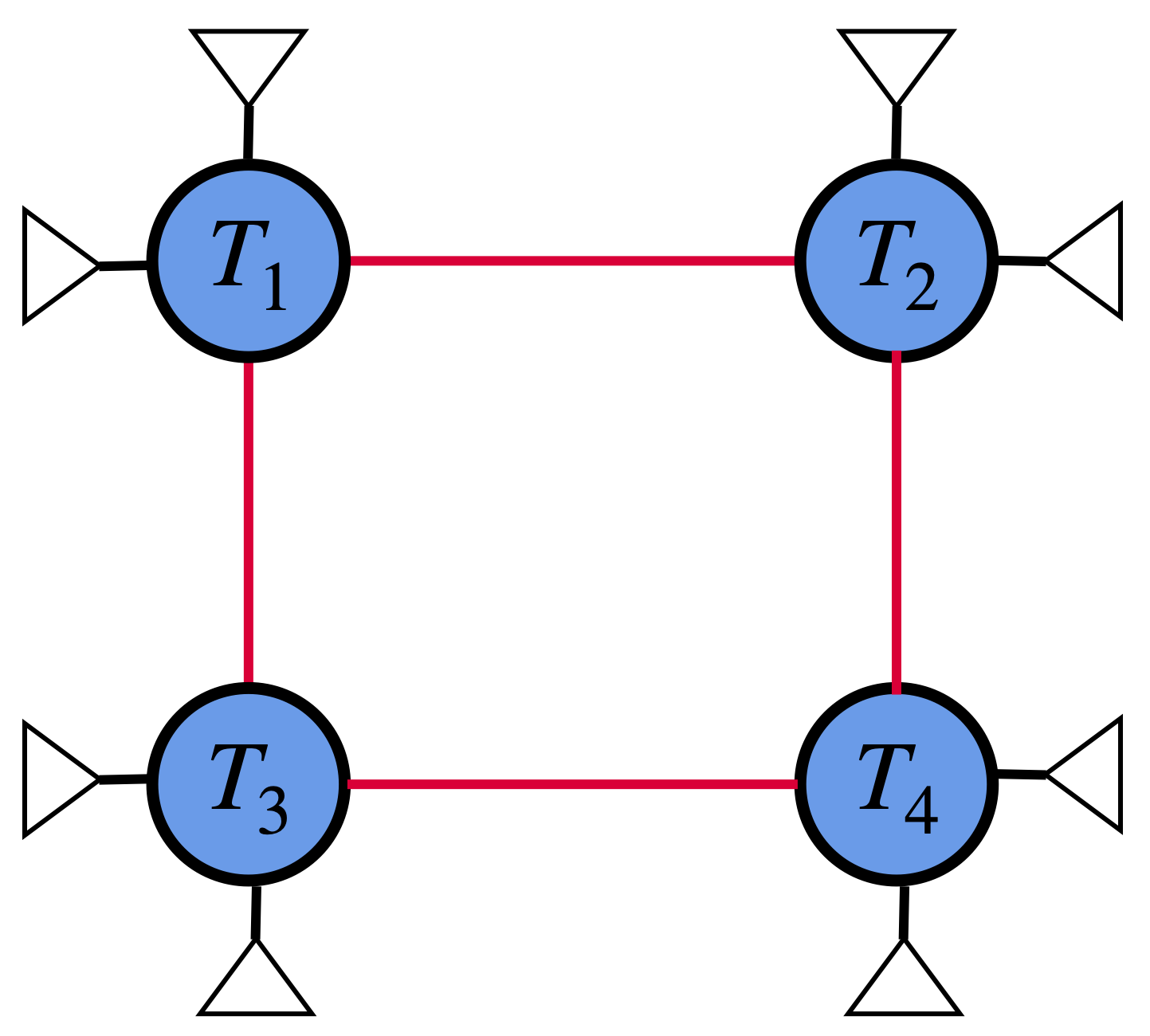} 
\end{figure}

where the message tensors from all other connected vertices are contracted in.

The loop series expansion is then given as follows.

\begin{lemma}[Loop series expansion] \label{lem:loop_expansion}
    Consider the expansion of the partition function of the tensor network $\TT = (\{T_v\}_{v\in V}, V, E)$ by resolving the identity at each edge. Then, we have that
    \begin{equation}
        \ZZ(\TT) = Z_0 +  \sum_{l\in \mathcal{L}_G} Z_{l}
    \end{equation}
    where, the only non-zero excited contributions are generalized loops $l$ in $G=(V,E)$.
\end{lemma}

We illustrate this expansion for a simple tensor network consisting of five vertices in Fig.~\ref{fig:fig_bp}. The net contraction is given as a sum of the BP vacuum along with all possible edge excitations, which appear as generalized loops.

\subsection{Divergence due to disconnected loops}

Ideally, one would like to approximate the loop expansion by truncating the series at some finite loop weight. Specifically, we set a cut-off weight and sum over loops with weights $\le m$. Denoting this by $\hat{Z}_m$, we have, 
\begin{equation}
    \hat{Z}_m = Z_0 + \sum_{\substack{ l\in \mathcal{L}_G \\ |l| \le m}} Z_l
\end{equation}

Empirically, one expects that if $Z_l$ decays exponentially fast with $l$, then $\hat{Z}_m$ provides a good approximation to $\mathcal{Z}(\mathcal{T})$ with a modest cutoff $m$. A major problem with this approach the need to sum over generalized loops that are disconnected. To see this, consider a tensor network (PEPS) defined on a $L \times L$ two-dimensional square lattice. On regular lattices such as the 2D square lattice, there are combinatorially many disconnected loops. For example, there are order $L^2 \choose
2$ disconnected loops of weight 8 formed by two loops of weight 4 on plaquettes (Fig. \ref{fig:disconnected_incompatibility}(a)). Going to higher weights, there are order $L^2 \choose 3$ disconnected loops formed by three connected loops, and so on. Therefore, as one goes to higher $m$, the number of terms grow combinatorially, outweighing the exponential decay of individual terms. This represents a significant bottleneck to the convergence of the loop series expansion, and motivates us to look for improved techniques. 
\section{Cluster expansion}\label{sec:cluster}

In this section, we introduce the cluster expansion and prove the main technical result: the tensor cluster expansion converges if loops $Z_l$ decay exponentially in $|l|$ with a sufficiently large exponent. Crucially, the cluster expansion technique overcomes the challenge of disconnected loops. We first give a physical picture about why cluster expansion provides a better series expansion than the loop series expansion. We then introduce the cluster expansion formally and state the main result. We give a toy example and compare with earlier work in the end.

\subsection{Physical intuition behind cluster expansion}
As we have argued in the previous section, tensor network contractions can be thought of as generalizations of partition functions in statistical mechanics by adding sign structures. Now, partition functions are not stable objects under local perturbations. For example, by heating any one site to infinite temperature, the partition function is changed by a constant \emph{multiplicative} factor. To account for this in a series expansion, each site must show up in a \textit{constant fraction} of terms. This is fundementally why the loop series necessitates the use of the disconnected loops and the combinatorial growth of the number of terms in the loop expansion of $\ZZ$.

On the other hand, the free energy $\FF$ is a well-behaved object under local perturbations. When one site is heated to infinite temperature, the free energy is changed only by a constant additive factor. This behavior is realizable in series expansions without disconnected objects: given a fixed site, only a \textit{constant number} of terms should be involved in the expansion. Therefore, one naively expects that the series expansion of the free energy is better behaved.

Another related intuition regarding the cluster expansion is the presence of non-linearity in the series to ensure convergence. As we will see, the non-linearity in the loop contributions added through the cluster method leads to provable convergence.

\subsection{Formal definition}

Now we formally introduce the cluster expansion. In particular, we employ the formalism of the abstract polymer model \cite{kotecky1986cluster} which provides black-box techniques to prove convergence. Throughout this subsection, we will work with the normalized tensor $\tilde{T}_v$ defined as follows.

\begin{equation}\label{eq:normalized_tensor}
    \tilde{T}_v = \frac{T_v}{Z^{(v)}},
\end{equation}
\noindent where $Z^{(v)}$ is the local contribution defined in Eq.~(\ref{eq:tensor_z}). Under this normalization, the BP contraction of $\tilde{T}_v$ is one, and correspondingly the BP free energy of $\tilde{T}_v$ is zero. We will compute the cluster expansion of $\FF(\tilde{T})$, which is related to $\FF(T)$ by a constant offset.

\begin{equation}
    \FF(T) = \FF(\tilde{T}) + \sum_v \log(Z^{(v)})
\end{equation}

\begin{figure*}[t]\label{fig:disconnected_incompatibility}   
    \centering
    \includegraphics[width=0.8\linewidth]{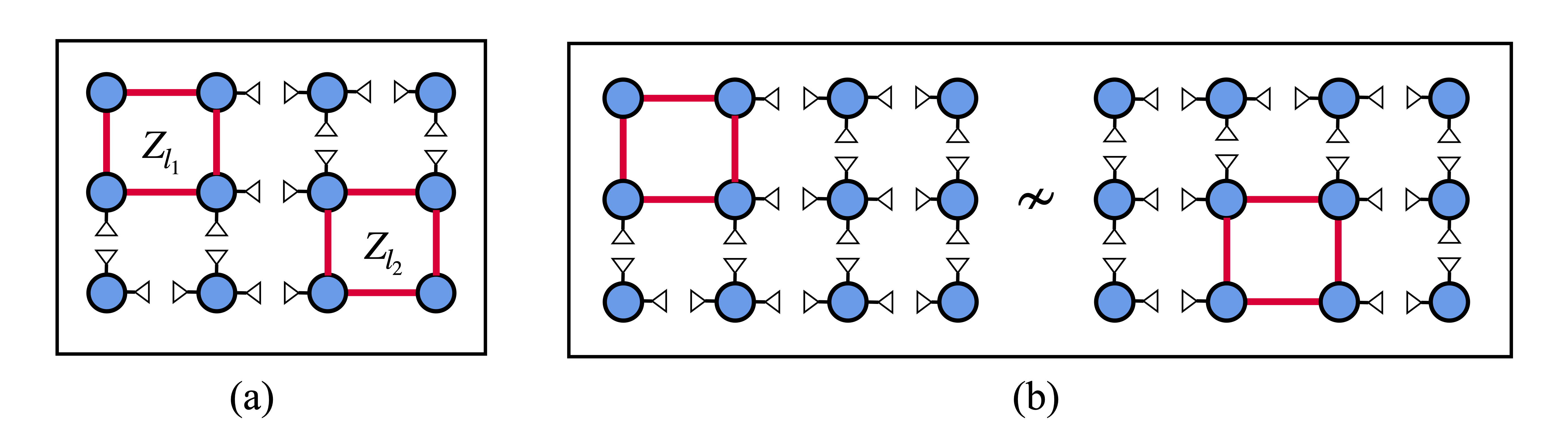}  
    \caption{(a) \textbf{Disconnectedness and incompatibility.} Example of a disconnected loop. (b) Example of two incompatible loops.}
    \label{fig:fig_cluster}  
\end{figure*}

Crucically, the loop series expansion of $\ZZ(\tilde{T})$ contains the contributions from all generalized loops, which includes connected as well as disconnected loops. We term the disconnected part as consisting of \emph{compatible} loops, in the following sense. Denote $\LL^c_G$ to be the set of connected loops in the graph $G$.

\begin{defn}[Compatible loops]\label{def:incompatibility}
Two loops \(l, l'\) are said to be \emph{compatible}, written \(l \sim l'\), if they do not overlap; that is, they share no vertex or edge in the underlying graph. A family \(\Gamma \subset \LL^c\) of loops is called \emph{compatible} if every pair of distinct loops in \(\Gamma\) is compatible.
\end{defn}
We give a pair of incompatible loops in Fig. \ref{fig:disconnected_incompatibility}(b). We note that the notion of \textit{loop compatibility} is conceptually similar, though not identical, to the notion of \textit{connectedness} of loops. Connectedness is a property of a single loop, describing whether it can be decomposed into two disconnected subgraphs. In contrast, loop compatibility is a relation between two distinct loops and does not depend on their individual connectedness. If two loops are compatible, their union forms a larger but disconnected loop.




Next, we define the object \emph{cluster}. Intuitively, clusters are a collection of loops, which can be described as a multiset as follows.

\begin{defn}[Clusters]\label{def:cluster}
  A cluster is a collection of tuples of the form $$\mathbf{W} = \{(l_1, \eta_1), (l_2, \eta_2), \ldots, (l_m, \eta_m)\}$$ where each \(l_i\in \LL\) is a loop and \(\eta_i\) is the multiplicity of the loop \(l_i\) in the cluster. The total number of loops in the cluster is denoted as $n_{\mathbf{W}} = \sum_{i=1}^m \eta_i$.
\end{defn}

We define the cluster weight $|\mathbf{W}| = \sum_{i} \eta_i |l_i|$, where $|l_i|$ is the weight of loop $l_i$. We also denote $\mathbf{W}! = \prod_i \eta_i !$. We denote the correction of the cluster $Z_{\mathbf{W}}$ as the product of the loop corrections raised to their respective multiplicities:
\begin{defn}[Cluster correction]
For a cluster \(\mathbf{W} = \{(l_1, \eta_1), (l_2, \eta_2), \ldots, (l_m, \eta_m)\}\), the cluster correction is defined as
\begin{equation}\label{eq:cluster_contribution}
  Z_{\mathbf{W}} = \prod_{i=1}^m Z_{l_i}^{\eta_i}.
\end{equation}
where $Z_{l_i}$ are the corresponding loop corrections.
\end{defn}

Given a cluster \(\mathbf{W}\), we define the \emph{interaction graph} as follows.

\begin{defn}[Interaction graph]\label{def:interaction_graph}
  Given a cluster \(\mathbf{W} = \{(l_1, \eta_1), (l_2, \eta_2), \ldots, (l_m, \eta_m)\}\), we define the interaction graph \(G_\mathbf{W} = (V_\mathbf{W}, E_\mathbf{W})\) with $|V_\mathbf{W}| = \sum_{i=1}^m \eta_i$ vertices with each loop \(l_i\) corresponds to \(\eta_i\) vertices. There is an edge $(l,l') \in  E_\mathbf{W}$ either if the loops \(l\) and \(l'\) are incompatible $l \not\sim l'$, or they are identical $l=l'$ 
\end{defn}

A cluster \(\mathbf{W}\) is called \emph{connected} if its interaction graph \(G_\mathbf{W}\) is connected—that is, there exists a path between any two vertices in \(G_\mathbf{W}\). This connectivity condition is crucial: we now establish that only connected clusters contribute to the free energy expansion.

\begin{lemma}[Connected clusters only]\label{lem:connected_cluster}
The free energy can be expressed as
\begin{equation}
  \FF(\tilde{T}) = \sum_{\text{connected} \, \mathbf{W}} \phi(\mathbf{W}) Z_{\mathbf{W}},
\end{equation}
where the sum runs over all connected clusters \(\mathbf{W}\). The coefficient \(\phi(\mathbf{W})\) is called the Ursell function, given as
\begin{equation}\label{eq:ursell}
  \phi(\mathbf{W}) =
  \begin{cases}
      1, &\begin{aligned}&\eta_{\mathbf{W}}=1 \end{aligned}\\
      \frac{1}{\mathbf{W}!} \sum_{\substack{C \in G_\mathbf{W} \\ C \text{ connected}}} (-1)^{|E(C)|},& \begin{aligned}&\eta_{\mathbf{W}} > 1\end{aligned}
  \end{cases}
\end{equation}
Where $C$ is a connected subgraph of the interaction graph $G_\mathbf{W}$ spanning all vertices, and $E(C)$ is the edge set of $C$.
\end{lemma}
The proof of Lemma \ref{lem:connected_cluster} is rather technical so we defer it to the appendix. Since the expansion involves only connected clusters, an important question arises: how many connected clusters must be enumerated in the truncated sum? The answer depends fundamentally on the graph structure, as quantified by the following combinatorial bound.

\begin{lemma}
    Given any graph with $n$ vertices and with degree $\Delta$, the number of connected clusters with weight $\le m$ is upper-bounded by $n (\Delta + 2)^m$
\end{lemma}

We prove this bound in Appendix \ref{app:polymer}. This bound reveals that there are $n \Delta^{O(m)}$ connected clusters per site, making the enumeration computationally tractable for moderate values of $m$. Given this bound, we are then motivated to \textit{truncate} the cluster series at a finite cluster weight $m$ to get the truncated free energy \(\tilde{F}_m\). 

\begin{defn}
[Truncated cluster expansion] The truncated partition function \(\tilde{F}_m\) is defined as
\begin{equation}
  \tilde{F}_m = \sum_{\substack{
    \text{connected }\, \mathbf{W} \\
    |\mathbf{W}| \le m
  }} \phi(\mathbf{W}) Z_{\mathbf{W}}.
\end{equation}
\end{defn}

We will use the \(\tilde{F}_m\) to approximate $\FF(\tilde{T})$. Our main technical result is to show that when the loop contribution decays exponentially with their weight with an exponent above a constant threshold, then the cluster expansion converges and \(\tilde{F}_m\) gives a good approximation to $\FF(\tilde{T})$.

\begin{theorem}[Convergence of the cluster expansion]\label{thm:convergence_informal}
(Informal) Given a tensor network with degree $\Delta$ and normalized by the BP fixed point. Assume there exists a constant \(c>c_0 = \Theta(\log(\Delta))\) such that
\begin{equation}
  |Z_l|\le e^{-c |l|}
\end{equation}
Then the series for \(\mathcal{F}\) converges absolutely, and the error in truncating the series at order \(m\) is bounded by
\begin{equation}
  \left| \mathcal{F} - \tilde{F}_m \right| = O(n e^{-d(m+1)})
\end{equation}
Where $d = c - c_0$.
\end{theorem}

We prove our main result in Appendix \ref{app:polymer} which is a direct application of the Kotechy-Preiss condition \cite{kotecky1986cluster}. If we denote the true free energy \textit{density} as $f = \FF/n$ and the weight$-m$ cluster approximation density as $\tilde{f}_m = \tilde{F}_m/n$, then we have that sufficient loop decay ensures,
\begin{equation}
  \left| f - \tilde{f}_m\right| = O(e^{-d(m+1)})
\end{equation}
We discuss the implications of our main theorem. At $m=0$, our theorem tells that BP approximates the free energy density up to a constant additive error, under the stated assumption. Further, cluster expansion improves this error exponentially fast in $m$. Hence, to get to an error $\epsilon$, one needs to enumerate clusters to order $m = \frac{1}{d}\log{\frac{1}{\epsilon}}$. Moreover, all connected clusters with weight $\le m$ can be enumerated in $n\, e^{O(m)} = O(n/\epsilon^{1/d})$ time. Thus, the time complexity is $\poly(n)$ to get to an inverse polynomial error in free energy density.

Finally, in certain cases such as simulating quantum dynamics, the tensor network contraction $\mathcal{Z}(\mathcal{T})$ itself is the physical observable, and we are interested in quantifying its error. An \emph{additive} error of $\epsilon$ in $\mathcal{F}(\mathcal{T})$ corresponds to a \emph{multiplicative} error of $\Theta(\epsilon)$ in $\mathcal{Z}(\mathcal{T})$. Since these observables are often of order $O(1)$, this typically implies an $O(1)$ additive error. However, when the observables become exponentially small, an $O(1)$ additive error is no longer meaningful. In contrast, a $\Theta(\epsilon)$ multiplicative error ensures that the additive error bar shrinks proportionally as the observable decreases. This makes the cluster expansion particularly favorable in such regimes.

\subsection{Toy example}\label{sec:toy}
To illustrate the idea behind the cluster expansion and compare it to the loop expansion, we consider a toy example. Consider a tensor network on a one-dimensional ring. The only generalized loop is the entire ring $l$, and the loop contribution is $Z_l$. Suppose we have normalized the tensor network by the BP fixed point, so the BP contribution is one. Then, the loop expansion gives
\begin{equation}
    \ZZ = 1 + Z_l
\end{equation}
and the free energy is
\begin{equation}
    \FF = \log(1 + Z_l)
\end{equation}
On the other hand, the only possible clusters are $\{(l,1)\}, \{(l,2)\}, \ldots$, namely the same loop repeated multiple times. In this case, all clusters of weight $m$ are $\{(l,m)\}$. 
Therefore, $\mathbf{W}!=m!$. The part that sums over connected spanning graphs evaluates to $ (-1)^{m+1} (m-1)!$ which we show in Appendix \ref{app:ursell_toy}. Therefore,  the Ursell function can be computed to be $\phi(\{(l,m)\}) = \frac{(-1)^{m+1}}{m}$. The cluster expansion gives
\begin{equation}
    \FF = \sum_{k=1}^\infty \frac{(-1)^{k+1}}{k} Z_l^k
\end{equation}
This is exactly the Taylor expansion of $\log(1+Z_l)$, which converges for $|Z_l|<1$.

We note the in this small example, loop expansion converges in the first order, while the cluster expansion needs to go to higher orders. However, when $|Z_l|$ is small, both methods agree on the leading order. As we will see later, when contracting large tensor networks, we expect cluster expansion to converge faster. We also note that while the computation of the Ursell function is daunting here, it is drastically simplified in large tensor networks since one typically does not need to handle clusters with many loops. In fact, in the numerical work in the next section, the Ursell function can be brute-force enumerated in the order we truncate.

Finally, this toy example also illustrates the difference between the linked cluster expansion and our cluster expansion. The linked cluster expansion produces $\log(1+Z_l)$ in one iteration \cite{welling2012cluster}, whereas our cluster expansion is an infinite series expansion.

\section{Algorithms}\label{sec:algorithm}

In this section, we present an overview of the algorithmic procedure for computing the cluster expansion of generic tensor networks. Figure~\ref{fig:pseudocode} provides a pseudocode summary. Suppose we are given a family of tensor networks $\{\mathcal T\}$ defined on a common graph $G$, and our goal is to contract each network. The algorithm takes as input the set $\{\mathcal T\}$ and a maximal cluster weight $m$, and outputs the truncated cluster expansion $\tilde{F}_m$ for each $\mathcal{T}$.

\subsection{Cluster Enumeration}
The first step is to enumerate all connected clusters with weight $\le m$ and save them. This step is computationally expensive, as its complexity grows exponentially with $m$. However, for a given graph $G$, this computation only needs to be performed once. The cluster enumeration algorithm is intricate; therefore, we provide a detailed discussion in Appendix \ref{app:algorithm}. Here, we summarize the main steps:
\begin{enumerate}
    \item For each vertex, enumerate all connected loops with weight $\le m$ supported on that vertex.
    \item Repeat over all vertices to obtain a list of connected loops, then deduplicate to remove redundancies.
    \item For each vertex, enumerate all connected clusters with weight $\le m$ supported on that vertex, using the list of connected loops.
    \item Repeat over all vertices to obtain a list of connected clusters, then deduplicate.
\end{enumerate}
For step 1, we use a depth- or breadth-first search algorithm to ``grow" a connected subgraph from each vertex, recording the subgraph only when it forms a generalized loop (see Definition \ref{def:loop}). Repeating this process over all vertices yields a list of connected loops. Since a single loop may be supported on multiple vertices, we perform de-duplication to remove redundancies. 

Step 3 follows a similar approach: starting from each vertex, we ``grow" connected clusters using the list of loops, and again deduplicate after repeating over all vertices. Finally, we repeat Step 3 for each vertex and deduplicate to obtain the final list of connected clusters.

Steps 1 and 3 have a runtime of $O(\exp(m))$, while steps 2 and 4 add an additional factor of $n$ \footnote{Technically, deduplication has a runtime of $O(n^2)$, but its constant factor is negligible.}. Thus, the total runtime is $O(n \exp(m))$. In practice, steps 1 and 2 are the most time-consuming. Appendix \ref{app:algorithm} discusses strategies to improve runtime, such as exploiting symmetries (such as translational symmetry) of the graph $G$.

\subsection{Message Passing and Normalization}
After enumerating all connected clusters, we proceed to compute the cluster expansion for each tensor network $\mathcal{T}$. The next step is to run BP on $\mathcal{T}$ to obtain the (approximate) fixed-point messages $\{\mu_{v \to s}\}$. This is achieved by iteratively applying the following update rule until convergence:
\begin{equation}
        \mu_{v\to s} \rightarrow \left(\underset{r\in \NN(v)/\{s\}}{\otimes}\mu_{r\to v} \right)\star T_v
\end{equation}
Schematically,
\vspace{-0.2cm}
\begin{figure}[H]
    \centering
    \includegraphics[width=5cm, keepaspectratio]{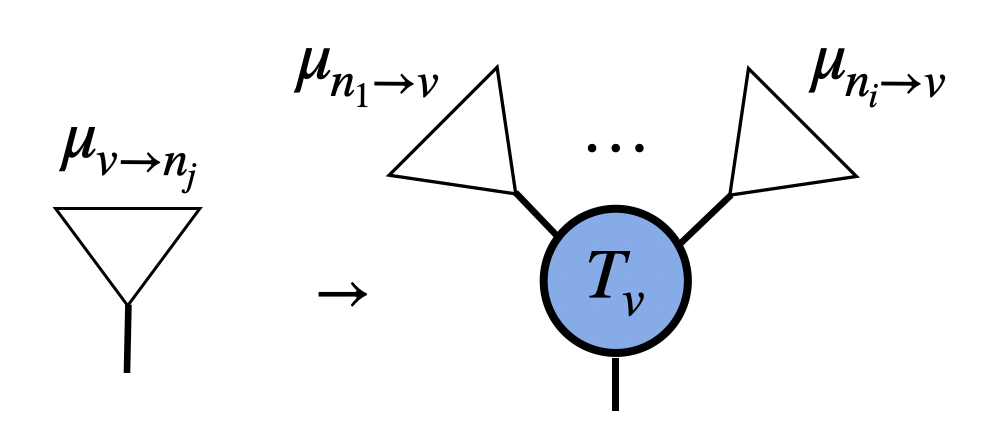} 
\end{figure}
\vspace{-0.2cm}

For numerical stability, we normalize the message vector by its two-norm after each update. This does not affect the result in the end as any constant factor will be cancelled out when dividing by $I_{vw}$. Message passing is typically efficient, as belief propagation often converges rapidly. However, tensor networks with sign structures may either fail to converge or converge slowly. Additionally, message passing can admit multiple fixed points, some of which may yield poor approximations.

Empirically, we find that the following strategies are helpful: (1) initializing messages randomly (in case of symmetry breaking), (2) \textit{damping} the message passing—--i.e., updating messages as a convex combination of the old and new values, and (3) adding a small amount of noise to the messages at each iteration. It is imperative to ensure that the self-consistency condition is satisfied to some fixed error $\varepsilon$ at each vertex, that is, $\forall v\in V, s\in \NN(v)$ we ensure,

\begin{equation}\label{eq:approxselfcons}
        \left\|\left(\underset{r\in \NN(v)/\{s\}}{\otimes}\mu_{r\to v} \right)\star T_v - \mu_{v\to s} \right\|_2 < \varepsilon
\end{equation}

Once the messages have converged, we compute the BP free energy $F_{0}$ and normalize the tensors $\mathcal{T}$ to $\tilde{\mathcal T}$ as in Eq.~\ref{eq:normalized_tensor}. This normalization introduces a constant offset in the free energy, so we add $F_0$ back to the final result.

\subsection{Computing Cluster Contribution and Final Result}
The final step is to load the list of connected clusters and compute their contributions $Z_{\mathbf{W}}$ as described in Eq.~\ref{eq:cluster_contribution}. This step is another computational bottleneck, since the number of clusters grows exponentially with $m$. However, it is highly parallelizable, as the contribution of each cluster can be computed independently. For graphs with significant symmetries (e.g., square lattices), many loops share the same shape, allowing them to be batched together for efficient graphical processing unit (GPU) acceleration.

Each cluster requires contracting relatively small loop tensors. When the bond dimension is small, this is efficient. However, when the bond dimension is large (which could happen in the context of simulating quantum dynamics), optimizing the contraction order is necessary. The Ursell function $\phi(\mathbf{W})$, defined in Eq.~\ref{eq:ursell}, can in principle be precomputed, but in practice it is often straightforward to evaluate. For example, in PEPS, the Ursell function is always $1$, $-1$, or $-1/2$ for weights up to twelve. Finally, we sum over all connected clusters with weight $\le m$ to obtain the truncated cluster expansion. The final result is given by $F_0 + \sum_{\mathbf{W}} \phi(\mathbf{W})Z_{\mathbf{W}}$.

\begin{figure} 
    \includegraphics[width=\linewidth]{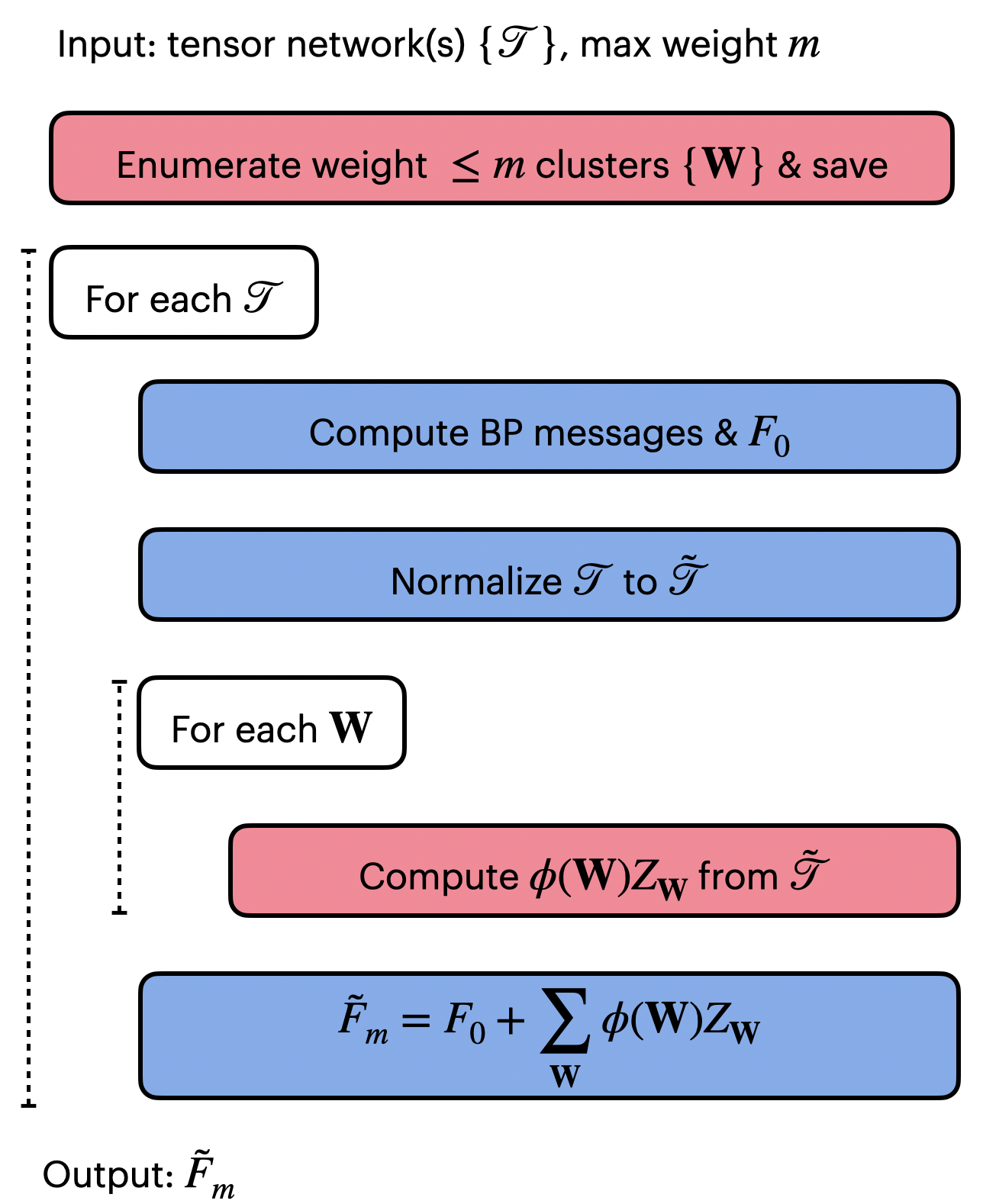}  
    \caption{\textbf{Pseudocode of Cluster Expansion}. Computationally expensive steps are colored in red. When there are more than one tensor networks, we assume they are defined on the same graph so that they share the same clusters.}
    \label{fig:pseudocode}  
\end{figure}
\section{Benchmark: 2D Ising Model}\label{sec:ising}

To demonstrate the efficacy of tensor network belief propagation and validate the cluster expansion methodology, we apply the framework to the paradigmatic two-dimensional classical Ising model. This system serves as an ideal testbed for several reasons: it possesses a known exact solution due to Onsager~\cite{onsager1944crystal}, exhibits a second-order phase transition with well-characterized critical behavior, and the partition function can be naturally formulated as a tensor network with constant bond dimension. The numerical calculations which follow have been performed using the \texttt{ITensor} \cite{fishman2022itensor} library.

We consider the classical 2D Ising model on a $L\times L$ square lattice with $N = L^2$ spins and nearest-neighbor interactions, as described by the Hamiltonian, 
\begin{equation}
    H[\{s_i\}] = -J\sum_{\langle ij\rangle} s_i s_j,
\end{equation}
where $s_i \in \{\pm 1\}$ denotes the spin variable at site $i$, $J > 0$ is the ferromagnetic coupling strength, and the sum runs over all nearest-neighbor pairs $\langle ij\rangle$ on a square lattice with periodic boundary conditions.  The partition function is given by
\[
Z = \sum_{\{s_i\}} \exp\left( \beta \sum_{\langle i,j \rangle} s_i s_j \right),
\]
where $\beta = 1/(k_B T)$ is the inverse temperature.

To represent this partition function as a tensor network, we use a factorization of the Boltzmann weights. For each nearest-neighbor interaction, we use the identity
\[
e^{\beta s_i s_j} = \sum_{x=0}^1 w(s_i, x, \beta) \, w(s_j, x, \beta),
\]
where the function $w(s, x, \beta)$ is defined as:
\[
w(s, x, \beta) =
\begin{cases}
\sqrt{\cosh(\beta)}, & x = 0, \\
\sqrt{\cosh(\beta)} \cdot s \cdot \sqrt{\tanh(\beta)}, & x = 1.
\end{cases}
\]

We then define a rank-4 local tensor $T_{x_1 x_2 x_3 x_4}$ at each lattice site, corresponding to the four directions (up, down, left, right), as follows:
\[
T_{x_1 x_2 x_3 x_4} = \sum_{s = \pm 1} w(s, x_1, \beta) \cdot w(s, x_2, \beta) \cdot w(s, x_3, \beta) \cdot w(s, x_4, \beta).
\]

The full partition function is then given by contracting these local tensors according to the 2D square lattice geometry:
\[
Z = \sum_{\{x_{i,j}\}} \prod_{\text{sites } (i,j)} T^{(i,j)}_{x_{i,j}^{(u)} x_{i,j}^{(d)} x_{i,j}^{(l)} x_{i,j}^{(r)}},
\]
where each index $x_{i,j}^{(\cdot)}$ is shared with the corresponding neighboring site, and the sum is over all internal bond indices $x_{i,j}^{(\cdot)} \in \{0, 1\}$, resulting in a bond dimension $\chi = 2$. The tensor network contraction of this ensemble produces the full partition function $\ZZ$, from which thermodynamic quantities such as the free energy density $f = -\beta^{-1} \ln(\ZZ)/N$ can be extracted. 

\subsection{BP Vacuum}
A fundamental question underlying any approximation scheme is understanding the regimes where it provides reliable results. For BP on tensor networks, this translates to identifying the physical conditions under which the BP vacuum accurately captures the system's behavior. Since BP implements a mean-field treatment—--approximating each site's environment with a rank-one tensor--—we expect it to perform well when mean-field assumptions are valid: deep within a phase and away from critical points.

Figure~\ref{fig:fig_ising}(a) tests this expectation by comparing the BP vacuum solution with Onsager's exact result for the free energy density $f(\beta)$ across the full temperature range. The BP approximation indeed demonstrates remarkable accuracy in both the high-temperature paramagnetic phase ($\beta \ll \beta_c$) and the low-temperature ferromagnetic phase ($\beta \gg \beta_c$), where deviations from the exact solution remain modest. However, significant discrepancies emerge in the critical region $\beta \in [0.25, 0.45]$ encompassing the phase transition, precisely where mean-field theory is expected to break down due to a diverging correlation length and enhanced fluctuations.

This behavior can be further understood through the theoretical foundations of the BP approximation. The BP vacuum solution effectively implements the Bethe approximation, treating the square lattice as a locally tree-like structure by neglecting loop correlations. This mean-field-like treatment captures the essential physics away from criticality, where local correlations dominate, but becomes increasingly inaccurate near the phase transition where long-range fluctuations and thereby loop effects become significant.

For the Ising model on a Bethe lattice with coordination number $z$, the critical point occurs at $\beta_c^{(z)} = 0.5 \log{\frac{z}{z-2}}$ \cite{baxter2016exactly}. Since the square lattice has $z=4$, the BP critical point is located at $\beta_{\text{BP}} = \frac{\log(2)}{2} \approx 0.347$, which we can verify through the divergence of the specific heat computed from the BP vacuum free energy. This BP critical point lies below the true Onsager critical point $\beta_c \approx 0.441$, explaining why BP accuracy deteriorates well before the actual phase transition.

\begin{figure*}[t]   
    \centering
    \includegraphics[width=1.0\linewidth]{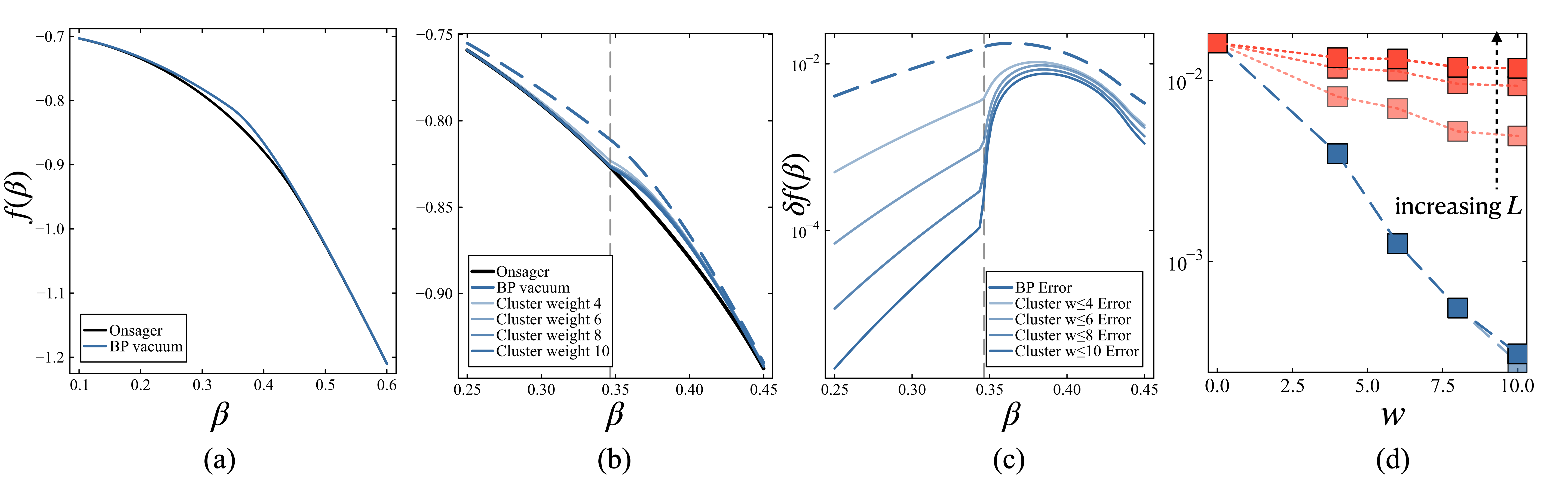}  
   \caption{\textbf{Ising Free Energy.} (a) Comparision of BP vacuum solution to the free energy density and the Onsager exact solution. (b) Cluster corrections to $f(\beta)$ in the critical region $\beta\in [0.25,0.45]$, with the BP critical point identified at $\beta_{\text{BP}} = \ln{2} /2$. (c) Free energy density error $\delta f(\beta)$ for BP vacuum (dashed) and different cluster corrections of weight $w\in\{4,6,8,10\}$. (d) $\delta f(\beta=\beta_{\text{BP}})$ for cluster corrections (dashed, blue) and loop corrections (dotted, red) as a function of cluster (loop) weight $w$ for system sizes $L \in \{10,20,30\}$. Curves of cluster expansion at different $L$ collapses because cluster expansion is automatically in the thermodynamic limit.}
    \label{fig:fig_ising}  
\end{figure*}

\subsection{Cluster Expansion}

To systematically improve upon the BP approximation, we implement the cluster expansion formalism by incorporating cluster corrections of increasing weight to the BP vacuum. Figure~\ref{fig:fig_ising}(b) presents a detailed view of free energy density in the critical region $\beta \in [0.25, 0.45]$, showing the progressive convergence toward the exact solution as cluster corrections of weight $w \in \{4, 6, 8, 10\}$ are added to the BP vacuum.

The cluster corrections exhibit distinct convergence behavior across different temperature regimes. In the high-temperature paramagnetic phase ($\beta < \beta_{\text{BP}}$), the corrections rapidly converge to the exact solution with relatively modest contributions from higher-order terms. However, in the low-temperature ferromagnetic phase ($\beta > \beta_{\text{BP}}$), convergence becomes markedly slower, requiring contributions from increasingly large clusters to achieve comparable accuracy. This difference reflects the degeneracy caused by the spontaneous symmetry breaking in the ferromagnetic phase. In the intermediate region, where the BP theory enters low-temperature but the 2D Ising model is still high-temperature, the choice of fixed points alters the convergence, which we detail in SM~Sec.~\ref{sec:add_num}.

The convergence properties of the loop expansion are further illuminated in Figure~\ref{fig:fig_ising}(c), which displays the free energy density error $\delta f(\beta) = f_{\text{approx}}(\beta) - f_{\text{exact}}(\beta)$ for the BP vacuum (dashed) and successive cluster corrections of weight $w \in \{4,6,8,10\}$. One notes the exponentially faster convergence as more cluster contributions are added for $\beta \leq \beta_{\text{BP}}$, and the bottleneck in convergence for $\beta > \beta_{\text{BP}}$.

\subsection{Convergence: Clusters v.s. Loops}

We now present the central numerical evidence of this work: a systematic comparison between our cluster expansion method and the `traditional' loop series expansion, revealing fundamental differences in how the algorithms scale. Figure~\ref{fig:fig_ising}(d) compares the approximation errors of cluster corrections (dashed, blue) and loop corrections (dotted, red) as functions of correction weight $w \in \{4,6,8,10\}$ for multiple system sizes $L \in \{10,20,30\}$.

The results reveal that while the cluster expansion exhibits robust exponential convergence that remains stable across all system sizes, the loop series expansion suffers from fundamental instabilities: (i) significantly slower convergence for any fixed system size, and (ii) divergent behavior as the system size increases, rendering it unsuitable for thermodynamic calculations.

This pathological behavior of loop expansions has a clear mathematical origin. Consider the BP vacuum contribution $Z_0 = z_0^N$ and a normalized weight-$w$ loop contribution $Z_w = O(1)$ that is intensive in the system size (normalized). The loop series expansion computes the free energy density as 
\begin{equation}
    \frac{1}{N}\log[Z_0(1 + NZ_w)] = \log z_0 + \frac{1}{N}\log(1 + NZ_w)
\end{equation}
In the thermodynamic limit $N \to \infty$, the correction term $\frac{1}{N}\log(1 + NZ_w) \to 0$, causing the loop contributions to vanish and negating any systematic improvement.

In contrast, our cluster expansion computes the free energy density as
\begin{equation}
    \frac{1}{N}[\log Z_0 + NZ_w] = \log z_0 + Z_w
\end{equation}
This is automatically in the thermodynamic limit in the sense that the estimated free energy density does not depend on the system size. This fundamental difference ensures that cluster corrections provide stable, size-independent improvements to the BP approximation, establishing the theoretical superiority of our tensor network-based approach for systematic corrections to belief propagation.

\subsection{Loop Contribution Analysis and Convergence Properties}

A fundamental question in the application of BP concerns the identification of parameter (temperature) regimes where such corrections on top of BP vacuum become most significant, thereby delineating the domains of validity for the BP vacuum approximation. To address this question, we analyze the average normalized loop contributions $Z_w(\beta)$ as functions of inverse temperature for loops of varying weight $w$.

Figure~\ref{fig:fig_ising_loop}(a) presents the temperature dependence of average loop contributions $Z_w$ for $w \in \{4, 6, 8, 10\}$. One notes that for all loop weights examined, the contributions exhibit a maxima precisely at the BP critical point $\beta_{\text{BP}} = \ln{2}/2$. This behavior provides compelling evidence that loop effects become most significant exactly where the BP approximation itself becomes critical, confirming the physical intuition that the breakdown of the tree-like approximation coincides with the emergence of strong loop correlations in the BP framework.

The observed peak structure has implications for the practical application of systematic corrections to belief propagation. Away from the BP critical point--—particularly in the high-temperature paramagnetic phase and the deep low-temperature ferromagnetic phase--—loop contributions remain relatively modest, indicating that the BP vacuum provides a robust zeroth-order approximation in thermodynamically stable phases. However, in the vicinity of $\beta_{\text{BP}}$, the amplification of loop effects signals the critical need for systematic inclusion of higher-order corrections to achieve quantitative accuracy. As we have demonstrated, this requirement is optimally addressed by the cluster expansion method, which exhibits exponential convergence and thermodynamic stability in precisely this challenging regime.

Equally crucial for validating our theoretical framework is examining the exponential decay rate of loop contributions $Z_w$ with increasing weight $w$—--this decay condition serves as the sufficient condition for convergence of our cluster expansion. Figure~\ref{fig:fig_ising_loop}(b) investigates this by analyzing the decay of loop contributions $Z_w(\beta)$ as a function of loop weight $w$ at three representative temperatures: the high-temperature phase ($\beta = 0.2$), the low-temperature phase ($\beta = 0.5$), and the BP critical point ($\beta_{\text{BP}} = \ln{2}/2$).

The results demonstrate that loop contributions decay exponentially with increasing loop weight across all temperature regimes, providing the essential sufficient condition for convergence of our cluster expansion series. In the high-temperature phase, the exponential decay is rapid and well-controlled, ensuring fast convergence reminiscent of traditional high-temperature expansions. However, the decay rate becomes notably slower in the ferromagnetic phase--—which could be related to spontaneous symmetry breaking. While loop magnitudes remain small in the ferromagnetic phase (confirming that the BP vacuum provides a good approximation there), the slower decay implies that cluster convergence becomes more challenging as one moves deeper into the ordered phase.

\begin{figure}[t]   
    \centering
    \includegraphics[width=1.0\linewidth]{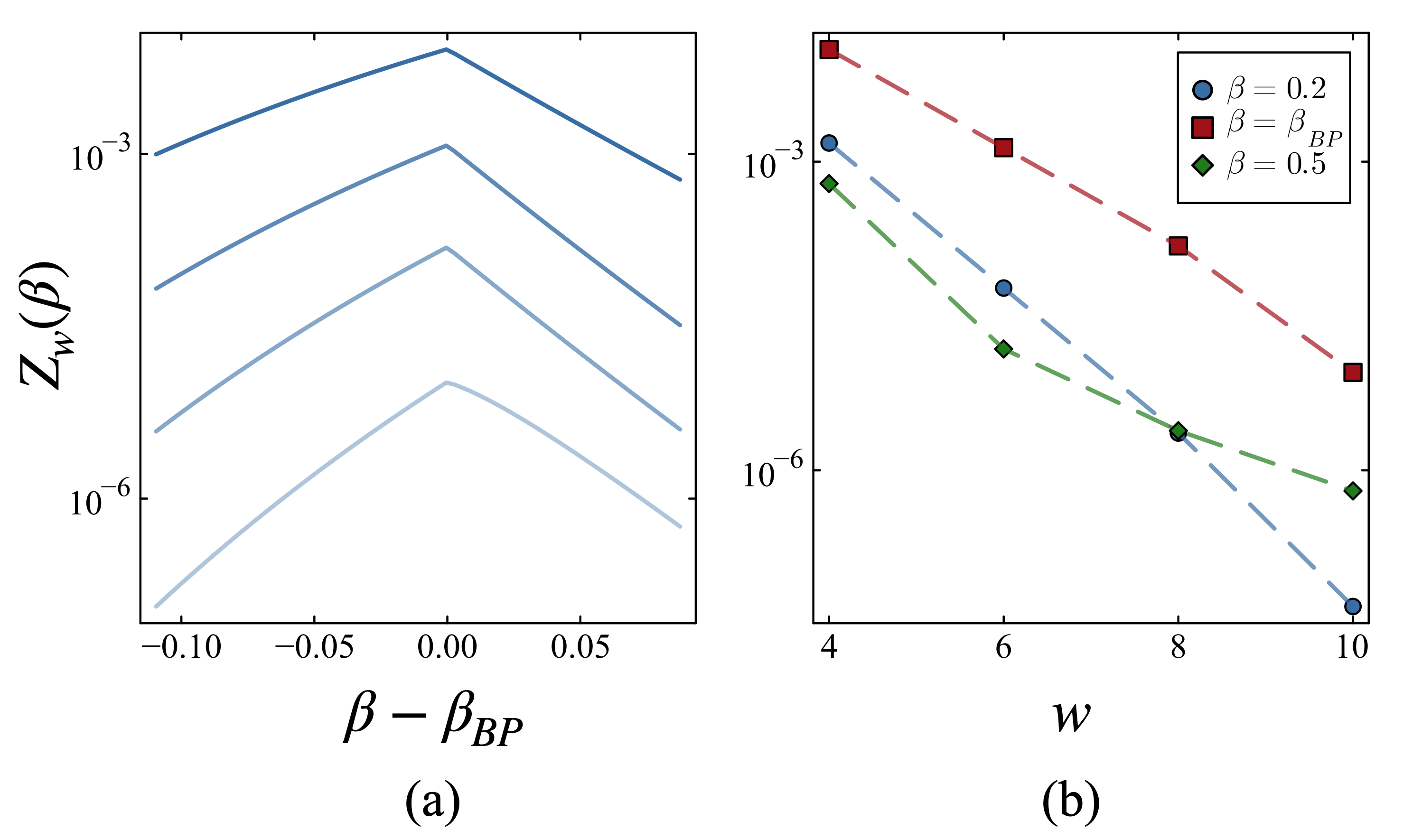}  
    \caption{\textbf{Ising Loop Decay}. (a) Average loop contributions $Z_w(\beta)$ for $w \in \{4,6,8,10\}$ showing a peak at the BP critical $\beta_{\text{BP}}$. (b) Decay of average loop contributions as a function of loop weight $w$ in the low-temperature phase $\beta = 0.5$, high-temperature $\beta = 0.2$ and critical point $\beta_{\text{BP}} = \ln{2}/2$.}
    \label{fig:fig_ising_loop}  
\end{figure}

\section{Discussion}
\label{sec:discussion}
We have presented a systematic theory of belief propagation for approximate tensor network contractions. Our main contribution is the construction of a cluster expansion that systematically improves upon the BP approximation. We rigorously prove that the cluster expansion converges exponentially fast if the loop contributions decay exponentially with a sufficiently large exponent. This result resolves two main challenges of BP: (1) it explains when BP provides a good approximation to the ground truth and (2) supplies an error estimate controlled by loop contributions. As a by product, it also yields a polynomial-time algorithm that systematically improves the BP result. 

To bring this technique into practice, we present a detailed and optimized algorithmic procedure for computing the cluster expansion. We benchmark our algorithm against both BP and the loop series expansion in the two-dimensional Ising model. Our results show that while BP deviates from the exact solution near the critical point, the cluster expansion yields significant improvements. Moreover, we demonstrate, both numerically and analytically, that the loop expansion fails to correct BP in the thermodynamic limit. Notably, this happens even at leading order before the onset of combinatorial growth.

While the cluster expansion does not converge in adversarial cases (otherwise there will be severe complexity-theoretic consequences), we believe that it should converge in many physically-relevant setups. In particular, the convergence of the cluster expansion reflects that the system is short-range correlated. Therefore, incorporating small clusters would already account for most of the correlations. This idea is not new in the literature; many earlier works have exploited short-ranged correlations to simply tensor network contractions, for example see \cite{lubasch2014algorithms,lubasch2014unifying}. An important theoretical question is to understand what kind of tensor networks admit a good BP approximation with a convergent cluster expansion.

Our work opens several avenues for future applications. First, BP is widely used in decoding classical and quantum error-correcting codes~\cite{mezard2009information,banchi2023generalization,mceliece1998turbo}. It would be interesting to explore the application of our method to improve the performance of BP-based decoders and provide rigorous error estimates. In particular, it is known that naive BP fails to give a threshold in quantum-LDPC codes because of the degeneracy problem~\cite{poulin2008iterative}. Since the cluster expansion captures the short-range correlations omitted by the BP approximation, it may help restore a threshold, if not simply improving the decoding performance, which we investigate in a forthcoming work \cite{BPTN_decoding}.

Second, many classical simulations of quantum systems and quantum dynamics heavily involve tensor networks, and more recently BP and its combination with other techniques~\cite{BP_gauging} in studying quantum dynamics. It would be interesting to explore the application of our method to improve the accuracy of these simulations. In particular, given well established competing methods such as TEBD, it is important to understand the regimes where BP is advantageous and can be incorporated into the simulators toolbox.

Finally, it would be interesting to extend our method to other tensor network geometries, such as higher-dimensional regular lattices or random expander graphs, where many conventional approaches break down. Belief propagation is particularly well-suited to these settings, as it is one of the few techniques that can handle arbitrary geometries while remaining computationally efficient. Moreover, BP is known to perform exceptionally well on locally tree-like structures, making it a promising tool for studying complex tensor networks beyond standard low-dimensional lattices.
\begin{acknowledgments}
We thank Hongkun Chen, Rhine Samajdar, David Huse, Sarang Gopalakrishnan, Dmitry Abanin, Joseph Tindall, Will Staples, and Dries Sels for helpful discussions and collaborations on related projects. We are particularly indebted to Grace Sommers for extensive discussions and cross-checking some of our numerical results. Y.F.Z acknowledges support from NSF QuSEC-TAQS OSI 2326767. The numerical simulations presented in this article were performed on computational resources managed and supported by Princeton University’s Research Computing. Parhey's company is also acknowledged.
\end{acknowledgments}

\section*{Note added}
After the completion of this work, a related work employing cluster expansion on tensor networks was posted \cite{gray2025tensor}. See Section \ref{sec:comparison} for detailed discussions.
\newpage 
\nocite{*}
\bibliographystyle{unsrt}
\bibliography{main}

\onecolumngrid
\renewcommand{\figurename}{SUPP FIG.}
\setcounter{figure}{0}    
\setcounter{section}{0}    

\newpage
\section*{Supplementary Material}

\section{Convergence via the Abstract Polymer Model}\label{app:polymer}

In this section, we establish the convergence of the loop expansion for the logarithm of the tensor network contraction $\log(\ZZ)$. The proof is based on the abstract polymer model which can be considered as a generalization of the cluster expansion in statistical mechanics. The convergence follows from the Kotecký–Preiss (KP) criterion~\cite{kotecky1986cluster}. Without loss of generality, consider the normalization wherein the BP vacuum contribution is unity. In other words, $\left [ \underset{w \in N(v)}{\bigotimes}\mu_{wv}\right] \star T_v = 1$ for each vertex $v$.

\subsection{Cluster Expansion}

We start from the loop expansion in Lemma \ref{lem:loop_expansion}. If a loop $l$ is a union of two disconnected loops $l_1$ amd $l_2$, then the loop contribution $Z_l$ factorizes into $Z_{l_1} \times Z_{l_2}$ (note that we normalize the BP contribution to one). Therefore, the loop series expansion for the normalized tensor network takes the following alternate form, 

\begin{prop}[Loop expansion reorganized]
The tensor network contraction admits the expansion
\begin{equation}\label{eq:loop_expansion_reorganized}
  \ZZ(\tilde{T}) = 1 + \sum_{\substack{\Gamma\subset\mathcal{L}\\\Gamma\text{ finite, compatible}}}
  \prod_{l\in\Gamma} Z_l
\end{equation}
where the sum runs over all finite sets \(\Gamma\) of mutually compatible loops.
\end{prop}

We now define the notion of cluster. A cluster is a multiset of loops.

\begin{defn}[Clusters]
  A cluster is a collection of tuples of the form $$\mathbf{W} = \{(l_1, \eta_1), (l_2, \eta_2), \ldots, (l_m, \eta_m)\}$$ where each \(l_i\in \LL\) is a loop and \(\eta_i\) is the multiplicity of the loop \(l_i\) in the cluster. The total number of loops in the cluster is denoted as $n_{\mathbf{W}} = \sum_{i=1}^m \eta_i$.
\end{defn}

Let the number of edges in loop $l$ be denoted as $|l|$. We define the cluster weight $|\mathbf{W}| = \sum_{i} \eta_i |l_i|$. We also denote $\mathbf{W}! = \prod_i \eta_i !$. We denote the correction of the cluster $Z_{\mathbf{W}}$ as the product of the loop corrections raised to their respective multiplicities:
\begin{defn}[Cluster correction]
For a cluster \(\mathbf{W} = \{(l_1, \eta_1), (l_2, \eta_2), \ldots, (l_m, \eta_m)\}\), the cluster correction is defined as
\begin{equation}
  Z_{\mathbf{W}} = \prod_{i=1}^m Z_{l_i}^{\eta_i}.
\end{equation}
\end{defn}

Given a cluster \(\mathbf{W}\), we define the \emph{interaction graph} as follows.

\begin{defn}[Interaction Graph]
  Given a cluster \(\mathbf{W} = \{(l_1, \eta_1), (l_2, \eta_2), \ldots, (l_m, \eta_m)\}\), we define the interaction graph \(G_\mathbf{W} = (V_\mathbf{W}, E_\mathbf{W})\) with $|V_\mathbf{W}| = \sum_{i=1}^m \eta_i$ vertices with each loop \(l_i\) corresponds to \(\eta_i\) vertices. There is an edge $(l,l') \in  E_\mathbf{W}$ either if the loops \(l\) and \(l'\) are incompatible $l \not\sim l'$, or they are identical $l=l'$ 
\end{defn}

We call a cluster \(\mathbf{W}\) \emph{connected} if the interaction graph \(G_\mathbf{W}\) is connected, meaning there is a path between any two vertices in the interaction graph.

We now show the most important lemma: only connected families of loops contribute to the expansion of \(\log \ZZ\). This is crucial for the convergence of the series.

\begin{lemma}[Connected clusters only]\label{lem:connected}
The free energy $\log\ZZ$ can be expressed as
\begin{equation}
  \log\ZZ = \sum_{m=1}^{\infty} \sum_{l_1 \in \LL^c} \sum_{l_2 \in \LL^c} \cdots \sum_{l_m \in \LL^c} \phi(l_1,\ldots,l_m) \prod_{i=1}^m Z_{l_i},
\end{equation}
with loops appearing including multiplicities. Define the cluster $\mathbf{W} = \{l_1, \ldots, l_m\}$, the coefficient $\phi(l_1,\ldots,l_m)$ is called the \emph{Ursell function} and is given by
\begin{equation}
  \phi(l_1,\ldots,l_m) =
  \begin{cases}
    \frac{1}{m!} \sum_{\substack{C \in G_\mathbf{W} \\ C \text{ connected}}} (-1)^{|E(C)|}  & \text{if $\mathbf{W}$ is connected} \\
    0 & \text{if $\mathbf{W}$ is disconnected.}
  \end{cases}
  \end{equation}
Where $C$ sums over all connected subgraphs of the interaction graph \(G_\mathbf{W}\) spanning all vertices, and \(|E(C)|\) is the number of edges in the subgraph \(C\). $i,j$ are vertices in $C$ and are implicitly mapped to the loops $l_i$ and $l_j$ in the ordered list $(l_1, \ldots, l_m)$. 
\end{lemma}

The above lemma sums over an ordered list of loops, which contains redundancies. We re-express the above lemma in terms of the cluster correction \(Z_{\mathbf{W}}\):
\begin{corollary}[Connected clusters only, reorganized]\label{lem:connectedclustersonly}
The free energy can be expressed as
\begin{equation}
  \log \ZZ = \sum_{\text{connected} \, \mathbf{W}} \phi(\mathbf{W}) Z_{\mathbf{W}},
\end{equation}
where the sum runs over all connected clusters \(\mathbf{W}\). The coefficient \(\phi(\mathbf{W})\) is given by
\begin{equation}
  \phi(\mathbf{W}) = \frac{1}{\mathbf{W}!} \sum_{\substack{C \in G_\mathbf{W} \\ C \text{ connected}}} \sum_{(i,j) \in C} (-1)^{|E(C)|}
\end{equation}
\end{corollary}

The proof of Corollary~\ref{lem:connectedclustersonly} follows from Lemma~\ref{lem:connected}, where we realize that each cluster \(\mathbf{W}\) shows up $m!/\mathbf{W}!$ times in the expansion of \(\log Z\). The factor $m!$ counts the permutations of the loops in the ordered list, while the factor \(\mathbf{W}!\) removes the redundancies due to the multiplicities of the loops in the cluster.

\subsection{Cluster expansion of the free energy}
We derive the cluster expansion of the free energy from the loop expansion of the partition function. This proves Lemma~\ref{lem:connected} and Corollary~\ref{lem:connectedclustersonly}. The proof follows from Chapter 5 of \cite{friedli2017statistical}.

\begin{proof}

(Proof of Lemma \ref{lem:connected}) We start from the loop expansion of the partition function (Eq.(\ref{eq:loop_expansion_reorganized})), reproduced here for convenience:

\begin{equation}
  \ZZ(\tilde{T}) = 1 + \sum_{\substack{\Gamma\subset\mathcal{L}\\\Gamma\text{ finite, compatible}}}
  \prod_{l\in\Gamma} Z_l
\end{equation}
Where the sum runs over all finite sets \(\Gamma\) of mutually compatible loops. Next, we convert the loop expansion of $\ZZ(\tilde{T})$ to a cluster expansion of the free energy $\FF(\tilde{T})$. For two connected loops $l_i$ and $l_j$, we define $\Delta(i,j)$ to be one if they are incompatible, and zero otherwise. Let the total number of loops be $|\mathcal{L}| = N$. We can rewrite Eq. (\ref{eq:loop_expansion_reorganized}) as
\begin{equation}
  \ZZ(\tilde{T}) = 1 + \sum_{m=1}^\infty \frac{1}{m!} \sum_{l_1, l_2, \ldots, l_m} \prod_{i} Z_{l_i} \prod_{1\le i<j\le m} (1 - \Delta(i,j))
\end{equation}
Where each $l_i$ sums over all connected loops. The factor of $1/m!$ removes the overcounting when going to an ordered list of loops. One can see that whenever there is a pair of incompatible $1 - \Delta(i,j)$ becomes zero. Next, we expand the product of $1 - \Delta(i,j)$ in the following way.
\begin{equation}
  \ZZ(\tilde{T}) = 1 + \sum_{m=1}^\infty \frac{1}{m!} \sum_{l_1, l_2, \ldots, l_m} \prod_{i} Z_{l_i} \sum_{G \in K_m} \prod_{(i,j) \in E(G)} (-\Delta(i,j))
\end{equation}
Where $K_m$ is the complete graph on $m$ vertices and $G$ sums over all subgraphs of $K_m$. $E(G)$ is the edge set of $G$. To simplify notations, for each graph $G$ we define $Q_G$ as
\begin{equation}
  Q_G = \sum_{l_1, l_2, \ldots, l_m} \prod_{a} Z_{l_a} \prod_{(i,j) \in E(G)} (-\Delta(i,j))
\end{equation}
Then we have
\begin{equation}
  \ZZ(\tilde{T}) = 1 + \sum_{m=1}^\infty \frac{1}{m!} \sum_{G \subset K_m} Q_G
\end{equation}
In general, $G$ may be disconnected. Suppose $G$ has $k$ connected components $G_1, G_2, \ldots, G_k$, then $Q_G$ admits the decomposition
\begin{equation}
  Q_G = \prod_{j=1}^k Q_{G_j}
\end{equation}
Plugging this back, we have
\begin{equation}
  \sum_{G \subset K_m} Q_G = \sum_{G \subset K_m}  \prod_{j=1}^k Q_{G_j}
\end{equation}
Instead of summing over $G$, we now sum over all possible partition of $m$ vertices into $k$ parts, and then sum over all connected graphs on each part. Therefore,
\begin{equation}
  \begin{split}
    \sum_{G \subset K_m} Q_G = \sum_{k=1}^{n} \sum_{\substack{m_1, m_2, \ldots, m_k \\ m_1 + m_2 + \ldots + m_k = m}} \frac{m!}{m_1! m_2! \ldots m_k!} \sum_{\substack{G_j \subset K_{m_j} \\ G_j connected, \forall j}}  \prod_{j} Q_{G_j}
  \end{split}
\end{equation}
Where $\frac{m!}{m_1! m_2! \ldots m_k!}$ counts the number of ways to partition $m$ vertices into $k$ parts with sizes $m_1, m_2, \ldots, m_k$. Instead of summing over $m$ first and then contraint $m_1 + m_2 + \ldots + m_k = m$, We can sum over $m_1, m_2, \ldots, m_k$ directly since $m$ goes to infinity. Therefore,
\begin{align}
  \ZZ(\tilde{T}) &= 1 + \sum_{k=1}^\infty \frac{1}{k!} \sum_{m_1, m_2, \ldots, m_k} \prod_{j=1}^k \left(\frac{1}{m_j!} \sum_{\substack{G_j \subset K_{m_j} \\ G_j connected}}  Q_{G_j} \right) \\
  &= 1 + \sum_{k=1}^\infty \frac{1}{k!} \left( \sum_{m} \frac{1}{m!} \sum_{\substack{G \subset K_{m} \\ G \, connected}}  Q_{G} \right)^k
\end{align}
We identify the second line as the exponential function, so we arrive at the cluster expansion of the free energy
\begin{align}
  \FF(\tilde{T}) = \log(\ZZ(\tilde{T})) = \sum_{m=1}^\infty \frac{1}{m!} \sum_{\substack{G \subset K_{m} \\ G connected}}  Q_{G}
\end{align}
Finally, we unwrap $Q_G$ to obtain the Ursell function and rephase the summation $G$ into a sum of clusters. In $Q_G$ we sum over $l_1, l_2, \ldots, l_m$ independently. Therefore, we group them into a cluster $\mathbf{W}$ and sum over all clusters. This incurs a factor of $m! / \mathbf{W}!$ to account for the overcounting. Next, we observe that the summation of over connected $G$ can be rewritten as follows.
\begin{equation}
  \sum_{\substack{G \subset K_{m} \\ G connected}} \sum_{(i,j) \in E(G)} (-\Delta(i,j)) = \sum_{\substack{G \, \text{spanning} \, G_{\mathbf{W}} \\ G \, \text{connected}}} (-1)^{|E(G)|}
\end{equation}
Where $G_{\mathbf{W}}$ is the interaction graph of the cluster $\mathbf{W}$ defined in Definition \ref{def:interaction_graph}. If $\mathbf{W}$ is disconnected, then no connected $G$ can span $G_{\mathbf{W}}$ so the summation is zero. Therefore, we arrive at
\begin{equation}
  \FF(\tilde{T}) = \sum_{m=1}^{\infty} \sum_{\substack{\mathbf{W} \text{with $m$ loops} \\ \mathbf{W} \text{ connected}}} \frac{1}{\mathbf{W}!} \left( \sum_{\substack{G \text{spanning} G_{\mathbf{W}} \\ G \text{ connected}}} (-1)^{|E(G)|} \right) Z_{\mathbf{W}}
\end{equation}
Identifying the coefficient as the Ursell function, we arrive at the final form of the cluster expansion of the free energy

\end{proof}

\subsection{Convergence via the Kotecký–Preiss Criterion}
We define the \emph{truncated partition function} \(\tilde{Z}_m\) as the sum of the contributions from clusters of weight at most \(m\):
\begin{defn}
[Truncated partition function] The truncated partition function \(\tilde{Z}_m\) is defined as
\begin{equation}
  \log \tilde{Z}_m = \sum_{\substack{
    \text{connected }\, \mathbf{W} \\
    |\mathbf{W}| \le m
  }} \phi(\mathbf{W}) Z_{\mathbf{W}}.
\end{equation}
\end{defn}

We now apply the Kotecký–Preiss criterion to show that the series for \(\log Z\) converges absolutely under certain conditions on the loop corrections \(Z_l\).
\begin{lemma}[Kotecký–Preiss criterion for the cluster expansion]\label{lem:KP}
  If there exists two constants $a$, $d$ such that for every loop \(l\) in the tensor network, we have
\begin{equation}\label{eq:KPcondition}
  \sum_{l':\;l'\not\sim l} |Z_{l'}|\,e^{(a+d)|l'|} \le a|l|,
\end{equation}
then the series for \(\log Z\) converges absolutely. Moreover, for any vertex $i$, we have the bound on the convergence:
\begin{equation}\label{eq:KPbound}
  \sum_{\substack{
    \text{connected }\, \mathbf{W} \\
    \mathbf{W} \not\sim i
  }} |\phi(\mathbf{W}) Z_{\mathbf{W}}| e^{\sum_{l \in \mathbf{W}} d|l|} \le a
\end{equation}
Where $\mathbf{W} \not\sim i$ denotes clusters supported on site $i$.
\end{lemma}

Let $|\mathbf{W}| = \sum_{l \in \mathbf{W}} |l|$ be the weight of the cluster, defined as the total number of edges in it. We truncate the sum over clusters to those of weight at most \(m\):



\begin{equation}
  \tilde{F}_m = \sum_{\substack{
    \text{connected }\, \mathbf{W} \\
    |\mathbf{W}| \le m
  }} \phi(\mathbf{W}) Z_{\mathbf{W}}.
\end{equation}


To apply the above lemma, we will set $a(l) = \frac{1}{2}|l|$, where $|l|$ is the number of edges in $l$. Next, we give a combinatorial estimate on the number of loops $l'$ that are incompatible with a given loop \(l\). The bound is a functtion of the size of the loop \(l\) and \(l'\), and the maximal degree of the tensor network, defined as the maximum number of legs of all tensors in the network.
\begin{lemma}[Combinatorial estimate on loops] \label{lem:combinatorial_estimate}
Let \(l\) be a loop of size \(k\) in a tensor network with maximum degree \(\Delta\). Then the number of loops \(l'\) of size \(m\) that are incompatible with \(l\) is bounded by
\begin{equation}
  N_m \le \frac{k}{2} (\Delta-1)^m.
\end{equation}
\end{lemma}
\begin{proof}
To begin with, the loop $l$ is supported on at most $k/2$ vertices, since each vertex has degree $>=2$. Therefore, we will bound the number of loop $l'$ supported on each of the $k/2$ vertices and then sum over all vertices.

For each vertex, the number of loops \(l'\) of size \(m\) supported on that vertex is bounded by the number of connected subgraphs of size \(m\) supported on that vertex. The latter is known to be bounded by \((\Delta-1)^m\) and is saturated by the tree. Therefore, the number of loops \(l'\) of size \(m\) that are incompatible with \(l\) is bounded by
\begin{equation}
  N_m \le \frac{k}{2} (\Delta-1)^m.
\end{equation}
\end{proof}

Applying Lemma~\ref{lem:combinatorial_estimate} to the Kotecký–Preiss criterion, we arrive at our main theorem on convergence.
\begin{theorem}[Convergence of the cluster expansion]\label{thm:convergence}
Given a tensor network with degree $\Delta$ and normalized by the BP fixed point. Assume there exists a constant \(c>\log(2(\Delta-1)) + \frac{1}{2} \) such that
\begin{equation}
  |Z_l|\le e^{-c |l|}
\end{equation}
Then the series for \(\log Z\) converges absolutely. Moreover, the error in truncating the series at order \(m\) is bounded by
\begin{equation}
  \left| \log Z - \tilde{F}_m \right| \le n e^{-d(m+1)}
\end{equation}
Where $d = c - \log(2(\Delta-1)) - \frac{1}{2}$.
\end{theorem}
\begin{proof}
By Lemma~\ref{lem:KP}, it suffices to show that
\begin{equation}
  \sum_{l':\;l'\not\sim l}|Z_{l'}|\,e^{(a+d)|l|} \le \frac{|l|}{2} \sum_{m\ge1} (\Delta-1)^{m} e^{(a+d-c) m} < a |l|
\end{equation}
where we have used the bound from Lemma~\ref{lem:combinatorial_estimate} and the assumption on the decay of \(Z_l\). This demands that
\begin{equation}
  \frac{1}{2} \sum_{m\ge1}  (\Delta-1)^{m} e^{(a+d-c) m} < a
\end{equation}
Set $r = (\Delta-1)e^{(a+d-c)}$. Then the series becomes
\begin{equation}
  \frac12\sum_{m\ge1}r^m = \frac12\frac{r}{1-r},
\end{equation}
and the condition
\(\tfrac12\sum r^m<a\)
implies that 
\begin{equation}
  \frac12\frac{r}{1-r} \le a
\end{equation}
We set $a= \frac{1}{2}$. This gives us the condition
\begin{equation}
  \frac{r}{1-r} \le 1
\end{equation}
This gives a condition on $d-c$:
\begin{equation}
  d-c \le -\log(2(\Delta-1)) - \frac{1}{2} 
\end{equation}
Since $d$ controls the rate of convergence, we want $d>0$. this gives the condition \(c>\log(2(\Delta-1)) + \frac{1}{2} \) for absolute convergence of the series.

Next we bound the error. Setting $a=\frac{1}{2}$ in Eq.(\ref{eq:KPbound}) and organizing the series by the weight of the clusters, we have
\begin{equation}
  \sum_{k=1}^{\infty} \sum_{\substack{
    \text{connected }\, \mathbf{W} \\
    \mathbf{W} \not\sim i \\
    |\mathbf{W}| = k
  }} |\phi(\mathbf{W}) Z_{\mathbf{W}}| e^{d k} \le \frac{1}{2}
\end{equation}
This implies that the series decays as at least \(e^{-d k}\). Therefore, truncating the cluster expansion at order \(m\) induces an error of $\frac{1}{2} e^{-d(m+1)}$.
\begin{equation}
  \sum_{k=m+1}^{\infty} \sum_{\substack{
    \text{connected }\, \mathbf{W} \\
    \mathbf{W} \not\sim i \\
    |\mathbf{W}| = k
  }} |\phi(\mathbf{W}) Z_{\mathbf{W}}| \le \frac{1}{2} e^{-d(m+1)}
\end{equation}
 Finally, We sum over all vertices \(i\) which over-counts the clusters. This results in the final bound
\begin{equation}
  \left| \log Z - \tilde{F}_m \right| \le \frac{1}{2}n e^{-d(m+1)}
\end{equation}
\end{proof}

Finally, we show that the number of connected clusters supported on one site grows at most exponentially. We restate the Lemma below.

\begin{lemma}[Number of Connected Clusters]\label{lem:num_connected_clusters_restate}
  The number of connected clusters \(\mathbf{W}\) of weight at most \(m\) supported on a vertex \(i\) is bounded by $(\Delta+2)^m$, where \(\Delta\) is the maximum degree of the tensor network.
\end{lemma}

\begin{proof}
    Starting from the graph $G$ of the tensor network, we define a new graph $G_e$ as follows. for each vertex $v$ in $G$ we create a series of vertices $v_1, v_2,\ldots$ in $G_e$. If there's a edge $(v,v')$ in $G$, we create edges between $(v_i, v'_i)$, for all $i$. We also create edges between $(v_i, v_{i+1})$, for all $v$ and $i$. One can think of $G_e$ as stacking layer of $G$ together and putting edges between layers.

    We show that every cluster $\mathbf{W}$ supported on vertex $v$ can be mapped to a connected subgraph $S_e$ of $G_e$ supported on $v_1$. Start by finding one loop $l_1$ in $\mathbf{W}$ that is supported on $v$. $l_1$ has to exist because otherwise $\mathbf{W}$ is not supported on $v$. Denote its edge set as $\{w^{(1)}, w^{'(1)}\}$. We add the edge set $\{w^{(1)}_1, w^{'(1)}_1\}$ (which is in the first layer) to $S_e$. 
    
    Next, we find another loop $l_2$ in $\mathbf{W}$ that is incompatible with $l_1$. $l_2$ has to exist because otherwise $\mathbf{W}$ is disconnected. Denote its edge set as $\{w^{(2)}, w^{'(2)}\}$. We add the edge set $\{w^{(2)}_2, w^{'(2)}_2\}$ (which is in the second layer) to $S_e$. 

We iterate this procedure. Every time we find a loop $l_i$ that is incompatible with one of the previously added loops and add it to some layer $j$. To determine $j$, we find the last layer $j_{\max}$ that contains an incompatible loop with $l_i$. We add the edge set $\{w^{(i)}_{j_{\max}+1}, w^{'(i)}_{j_{\max}+1}\}$ to $S_e$. None of the edges from this set has been added to $S_e$ because otherwise $l_i$ is incompatible with some loop in layer $j+1$. $S_e$ remains connected throughout this process: $l_i$ in layer $j_{\max}+1$ is connected to some loop in layer $j_{\max}$.
    
Therefore, we have mapped each cluster to a connected subgraph of $S_e$. We bound the number of connected clusters with the number of connected subgraphs supported on $v_1$. Since $S_e$ has a degree of $\Delta+2$, the tree bound of the connected subgraphs gives $(\Delta+2)^m$.
\end{proof}

\subsection{Evaluating the Ursell function of the Toy example}\label{app:ursell_toy}
In this section, we evaluate the Ursell function of the cluster $\{(l,m)\}$ in the toy example in the main text. First note that $\mathbf{W}!=m!$. Next, we compute the term that sums over connected graphs, which we call $C_m$.
\begin{equation}
    C_m = \sum_{\substack{G \, \text{spanning} \, G_{\mathbf{W}} \\ G \,\text{connected}}} (-1)^{|E(G)|}
\end{equation}
We first note that $G_{\mathbf{W}}$ is the complete graph with $m$ vertices since a loop is compatible with itself.  We denote the complete graph with $m$ vertices as $K_m$. We first define the following auxiliary quantity $A_m$ that sums over all spanning graphs that could be disconnected.
\begin{equation}
    A_m = \sum_{\substack{G' \, \text{spanning} \, G_{\mathbf{W}}}} (-1)^{|E(G)|}
\end{equation}
One can see that $A_m = (1 + (-1))^{|E(K_m)|}$ since $G_{\mathbf{W}}$ is the complete graph. Therefore, $A_m = 1$ when $m=0,1$ and $A_m = 0$ when $m>1$.

Now we relate $A_m$ to $C_m$. For each $G'$ we choose vertex one and let $S$ be the connected subgraph of $G$ that contains vertex one. Let $|S|=m'$. There are $\binom{m-1}{m'-1}$ possible choices of vertices in $S$. Since $S$ is disconnected with the complement $G / S$, $(-1)^{|E(G)|} = (-1)^{|E(S)|} \times (-1)^{|E(G/S)|} $. We sum over all possible vertex subsets forming $S$ to get
\begin{equation}
    A_m = \sum_{m'=1}^{m} \binom{m-1}{m'-1} C_{m'} A_{m-m'}
\end{equation}
We set $m>1$ so that $A_m=0$. The only non-trivial contribution on the right-hand side is when $m-m'=0,1$. Therefore,
\begin{equation}
    0 = \binom{m-1}{m-1} C_{m} A_0 + \binom{m-1}{m-2} C_{m-1} A_1 = C_m + (m-1) C_{m-1}
\end{equation}
Thus, we obtain a recursive relation for $C_m$. 
\begin{equation}
    C_m = -(m-1) C_{m-1}
\end{equation}
Starting from $C_1 = 1$, we get $C_m = (-1)^{m-1} (m-1)!$.

\section{Algorithm}\label{app:algorithm}

In this section we discuss the enumeration of connected clusters in detail. The first step is to enumerate all connected loops up to weight $m$ supported on a given vertex. We do this by starting from the given vertex and ``growing'' a connected subgraph by adding edges one at a time, until we reach the weight limit. This can be done using breadth-first search (BFS) or depth-first search (DFS). During the search, we check if the current subgraph is a generalized loop using \texttt{isGeneralizedLoop}, and if so we add it to the list of found loops. We also use \texttt{canPruneEarly} to discard branches that cannot lead to any valid generalized loop, which significantly speeds up the search. Empirically, we find that BFS is slightly faster than DFS, which could be related to the early pruning. The BFS algorithm is summarized in Algorithm \ref{alg:bfs-loop-m}.

We describe the functions \texttt{isGeneralizedLoop} and \texttt{canPruneEarly} here. Both function returns a true or false. \texttt{isGeneralizedLoop} checks if a connected subgraph is a generalized loop by verifying that all vertices have degree at least 2. Multiple conditions can enter \texttt{canPruneEarly}. First, if the number of degree-1 vertices exceeds twice the remaining edge budget, we can prune since each added edge can reduce the number of degree-1 vertices by at most 2. In addition, for graphs with symmetries (e.g., square lattices), differnt subgraphs can be related by symmetry operations (e.g., 90-degree rotations and reflections in square lattices). Therefore, if the current subgraph is related to a previously seen subgraph by a symmetry operation, we can prune the branch.

There could be redundancies in the found loops, e.g., the same loop could be found by different ways of growing. Therefore, we remove redundancies by using a canonical representation (e.g. sorted edge list) of the edge set to de-duplicate. The canonical representation can also be used in \texttt{canPruneEarly} to prevent revisiting the same subgraph.

\begin{algorithm}[H]
\textbf{Input:} Graph $G=(V,E)$, vertex $v\in V$, maximum weight $m$. \\
\textbf{Output:} $\mathcal{L}$, all connected subgraphs containing $v$ with $\lvert E(S)\rvert \le m$ that satisfy \texttt{isGeneralizedLoop}.
\caption{Enumerate generalized loops with BFS up to weight $m$}\label{alg:bfs-loop-m}
\begin{algorithmic}[1]
\State Initialize an empty list $\mathcal{L}$
\State Start from the trivial subgraph $S$ containing only $v$
\State Initialize a queue $Q$ and push $S$ into it
\While{$Q$ is not empty}
  \State Pop a subgraph $S$ from $Q$
  \If{\texttt{isGeneralizedLoop}$(S)$}
    \State Add $S$ to $\mathcal{L}$
  \EndIf
  \If{$|E(S)| = m$}
    \State \textbf{continue} \Comment{stop expanding if weight limit reached}
  \EndIf
  \For{each edge $e$ touching $S$}
    \State Form a new subgraph $S'$ by adding $e$ (and its endpoint) to $S$
    \If{\texttt{canPruneEarly}$(S')$}
      \State skip this branch
    \Else
      \State push $S'$ into $Q$
    \EndIf
  \EndFor
\EndWhile
\State \Return $\mathcal{L}$
\end{algorithmic}
\end{algorithm}

Next, we iterate over all sites in the graph, and for each site we find all connected loops supported on that site using Algorithm \ref{alg:bfs-loop-m}. If the graph has translation invariance, then we only run Algorithm \ref{alg:bfs-loop-m} on one site and translate the found loops to other sites. In the case where the bulk of the graph is translation invariant but there are boundaries, we run Algorithm \ref{alg:bfs-loop-m} on one bulk site and translate. When a loop runs over the boundary, we discard it. In the end we de-duplicate again since the same loop could grown from multiple sites.

With a list of all connected loops up to weight $m$, we can construct the loop interaction graph $F=(\mathcal{V},\mathcal{E})$ where each vertex $u\in\mathcal{V}$ is a connected loop. There is an edge between two vertices $u,u'\in\mathcal{V}$ if the corresponding loops are incompatible. Note a loop is always incompatible with itself so there is always a self-edge associated with each vertex. Then, we enumerate all connected clusters up to weight $m$ supported on a given site using a similar strategy. We start from each loop supported on the given site, and grow a connected cluster by adding neighboring loops in the interaction graph $F$. This step is usually much faster than enumerating loops, so we do not perform early pruning. We give a DFS version of the algorithm in Algorithm \ref{alg:dfs-clusters-m}.

\begin{algorithm}[H]
\textbf{Input:} Loop interaction graph $F=(\mathcal{V},\mathcal{E})$ where each $u\in\mathcal{V}$ is a generalized loop with weight $w(u)$; site $s$; maximum cluster weight $m$. \\
\textbf{Output:} $\mathcal{C}$, all connected clusters with weight $\le m$, supported on site $s$.
\caption{Enumerate connected loop-clusters (DFS, multiset, weight cap $m$)}\label{alg:dfs-clusters-m}
\begin{algorithmic}[1]
\State Initialize an empty list $\mathcal{C}$
\State Find all generalized loops supported on site $s$; call this set $\mathcal{S}\subseteq \mathcal{V}$
\For{each seed loop $u_0\in \mathcal{S}$}
  \State Start a cluster $C$ containing only $u_0$ (with multiplicity $1$)
  \State $\text{W} \gets w(u_0)$
  \State \textbf{DFS}$(C,\text{W})$
\EndFor
\State \Return $\mathcal{C}$
\\
\Function{DFS}{$C,\text{W}$}
  \State Add the current cluster $C$ to $\mathcal{C}$ \Comment{it is connected by construction and has total weight $\text{W}\le m$}
  \State Let $\partial C$ be all loop-vertices $v\in\mathcal{V}$ that are adjacent in $F$ to at least one loop appearing in $C$ (neighbors in $F$)
  \For{each $v \in \partial C$}
    \If{$\text{W} + w(v) \le m$}
      \State Form $C'$ by adding one more copy of $v$ to the multiset $C$
      \State \textbf{DFS}$(C',\, \text{W} + w(v))$
    \EndIf
  \EndFor
\EndFunction
\end{algorithmic}
\end{algorithm}

Finally, we iterate over all sites in the graph and run Algorithm \ref{alg:dfs-clusters-m}. We then perform deduplication again since the same cluster could be grown from multiple sites. This step can also exploit translation invariance if present.

\section{Additional Numerics} \label{sec:add_num}
\subsection{Fixed points of the Ising tensor}

A crucial aspect of the cluster expansion method is the choice of BP fixed point around which to perform the expansion. While the main text focuses on the algorithmic procedure, the selection of an appropriate fixed point fundamentally determines the convergence properties and accuracy of the subsequent cluster corrections. 

The BP algorithm and subsequent cluster corrections both operate by first establishing a set of self-consistent messages $\MM$. Typically, the fixed point is found through the iterative message passing procedure. However, for complex systems with multiple fixed points, the choice of which fixed point to use as the expansion basis becomes critical. 

In general, the fixed point landscape is uncontrolled and it is unclear whether there are more than one (or, zero) fixed points. However, for the 2D Ising model (or any translational invariant tensor network), we can map out the complete fixed point landscape as $\beta$ varies, revealing how this choice impacts the effectiveness of our method.

To characterize the fixed point structure, we define an ``energy" functional that measures self-consistency. Given the four-legged tensor $T$ and a message $\mu$, assuming translational invariance, we define:
\begin{equation}
    E(\mu) := \|\mu -  (\mu\otimes\mu\otimes \mu)\star T\|_2
\end{equation} 

The intuition behind $E(\mu)$ is that fixed points satisfying the self-consistency will have zero energy. For the Ising case, we parameterize $\mu \equiv \mu(\theta) = (\cos\theta,\sin\theta)$ and analyze $E(\mu(\theta)) \equiv E(\theta)$ as a function of angle $\theta$ in Supp. Fig.~\ref{fig:fig_fixed_points}(a). The results reveal a clear bifurcation structure: in the high-temperature regime $\beta < \beta_{\text{BP}}$, there exists a unique stable fixed point corresponding to the infinite-temperature solution. However, for $\beta > \beta_{\text{BP}}$, additional low-temperature fixed points emerge, creating multiple possible expansion bases.

This multiplicity of fixed points raises the following question: which fixed point provides the optimal basis for cluster expansion? To address this, we compare the performance of cluster expansions built upon different fixed points. Supplementary Figure~\ref{fig:fig_fixed_points}(b) shows the free energy density error for cluster expansions based on the message-passing fixed point (blue, solid) versus the infinite-temperature fixed point (green, dashed). The comparison reveals three distinct regimes with different optimal strategies: In the high-temperature phase ($\beta < \beta_{\text{BP}}$), both approaches yield identical results since the message-passing procedure naturally converges to the infinite-temperature fixed point. In the intermediate regime--—where BP theory predicts low-temperature behavior but the actual 2D Ising system remains in its high-temperature phase—--the low-temp fixed point provides a lower BP error, however the cluster expansion converges faster for the high-temp fixed point. This phenomenon reflects the fact that our BP+cluster method reduces to the traditional high-temperature cluster expansion when built upon the infinite-temperature fixed point. Finally, in the genuinely low-temperature phase, the message-passing fixed point significantly outperforms the infinite-temperature expansion, demonstrating the power of BP technology to capture the appropriate physics through the adaptive fixed point selection via message passing.

\begin{figure*}[t]   
    \centering
    \includegraphics[width=1.0\linewidth]{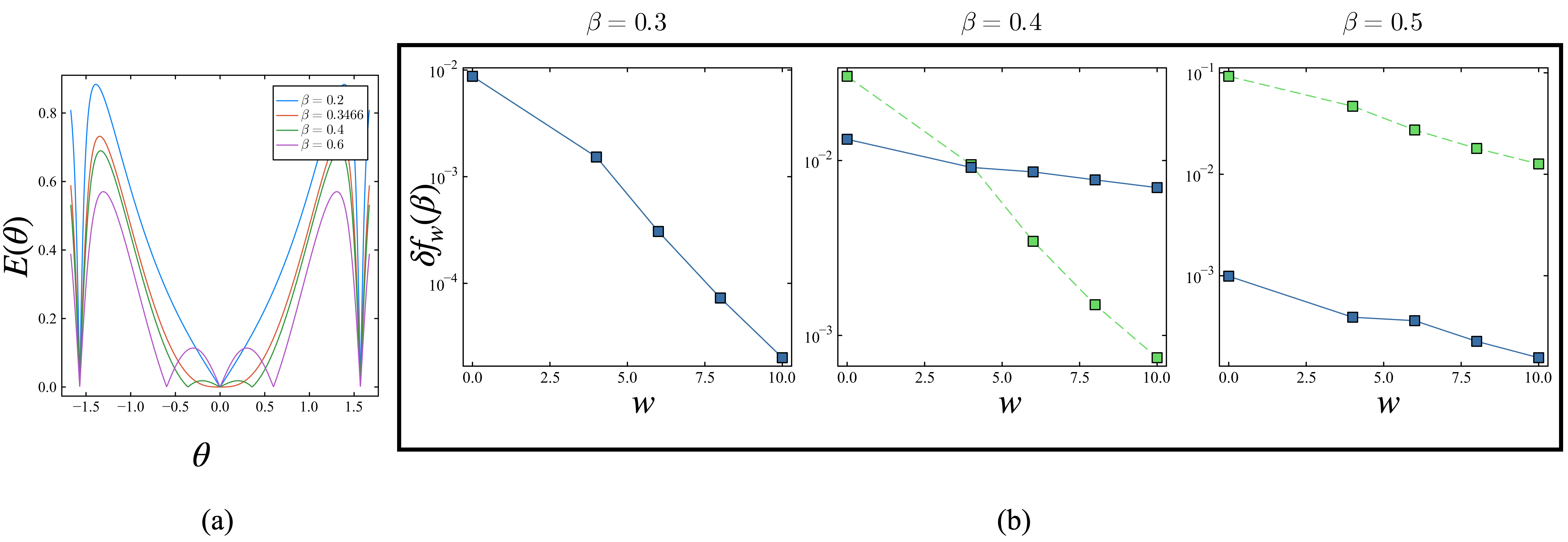}  
    \caption{\textbf{Fixed points:} (a) Fixed point `energy' landscape for the message $\mu_\theta = (\cos\theta,\sin\theta)$ (b) Free energy density error $\delta f_w(\beta)$ for the fixed point from message passing dynamics (blue, solid) and the infinite-temperature fixed point ($\theta=0$, green, dashed) for $\beta \in \{0.3,0.4,0.5\}$ and system size $L=20$.}
    \label{fig:fig_fixed_points}  
\end{figure*}

\subsection{BP correlation length and message propagation}

Having established that BP theory successfully captures distinct physics within different phases while predicting its own critical point, a natural question arises: how are correlations encoded and transmitted through the BP message structure? Understanding this mechanism is crucial for comprehending both the strengths and limitations of our cluster expansion approach.

To investigate the correlation structure within BP, we probe the system's response to localized perturbations. Specifically, we introduce a localized $z$-field perturbation at a single site in the 2D Ising model and examine how this disturbance propagates through the message network. This setup allows us to address a fundamental question: how does a local perturbation affect the self-consistent messages at distant locations?

Our analysis compares two systems: the original tensor network $\TT$ and its perturbed counterpart $\TT'$. After running the message-passing procedure to convergence on both networks, we quantify the perturbation's influence by measuring for each edge $e$,
\begin{equation}
    \Delta(e) = \| \mu_{e} - \mu'_{e}\|_2
\end{equation}
where $\MM = \{\mu_e\}$ and $\MM'=\{\mu'_e\}$ are the converged self-consistent messages for the unperturbed and perturbed systems, respectively. 

Supplementary Figure~\ref{fig:fig_message_correlation} displays this message difference on a logarithmic scale across the phase transition, examining inverse temperatures $\beta\in \{0.2,0.3,\beta_{\text{BP}}, 0.4,0.5\}$. The results reveal a striking temperature-dependent correlation structure: Deep within both high- and low-temperature phases, the perturbation's influence remains localized within an $O(1)$ neighborhood around the perturbation site, indicating a finite `message-correlation' length. However, at the BP critical point $\beta_{\text{BP}}$, the perturbation's influence extends over distances $O(L)$, propagating throughout the entire system. This critical behavior defines an effective BP correlation length $\xi_{\text{BP}}$ that diverges at the BP transition.

This correlation length analysis provides important insights into the computational complexity of our method. Through a simple light-cone argument, the message-passing algorithm's convergence time scales as $O(\xi_{\text{BP}} \cdot n)$, where $n = L^2$ represents the total number of vertices. 

\begin{figure*}[t]   
    \centering
    \includegraphics[width=1.0\linewidth]{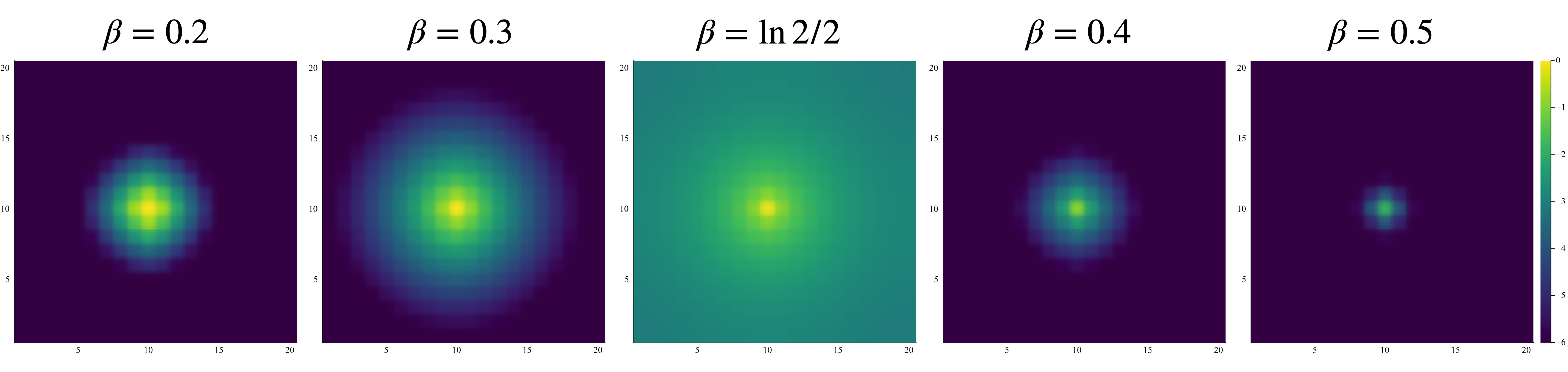}  
    \caption{\textbf{Message correlation:} Effect of a localized $z-$perturbation on the fixed point messages (plotted in log-scale) for the Ising model with system size $L=20$ across the phase transition, $\beta \in \{0.2,0.3,\beta_{\text{BP}}, 0.4,0.5\}$}
    \label{fig:fig_message_correlation}  
\end{figure*}

\end{document}